\documentclass[9pt,twocolumn,twoside]{osajnl}

\journal{arxiv} 

\setboolean{shortarticle}{false}

\usepackage{lineno}
\usepackage{animate}
\usepackage{setspace}
\usepackage{upgreek}
\usepackage{enumitem}
\usepackage{fixltx2e}
\usepackage{subcaption} 
\usepackage{capt-of}
\usepackage{booktabs}

\usepackage{array}
\usepackage{float}
\usepackage{multirow}
\usepackage{mathtools}
\usepackage{makecell}
\usepackage{cuted}
\usepackage{mathdefs}
\usepackage{my_macros}
\usepackage{weiyu_figure_macros}

\title{Enhancing Speckle Statistics for Imaging Inside Scattering Media}

\author[1]{Wei-Yu Chen}
\author[1]{Matthew O’Toole}
\author[1]{Aswin C. Sankaranarayanan}
\author[2,*]{Anat Levin}

\affil[1]{Carnegie Mellon University, 5000 Forbes Avenue Pittsburgh, PA 15213}
\affil[2]{Technion, Israel, Technion City, Haifa 3200003, Israel}

\affil[*]{Corresponding author: anat.levin@ee.technion.ac.il}

\begin{abstract}
We exploit memory effect correlations in speckles for the imaging of incoherent fluorescent sources behind scattering tissue.
These correlations are often weak when imaging thick scattering tissues and complex illumination patterns, both of which greatly limit the practicality of associated techniques.
In this work, we introduce a spatial light modulator between the tissue sample and the imaging sensor and capture multiple modulations of the speckle pattern.
We show that, by correctly designing  the modulation patterns and the associated reconstruction algorithm, the statistical correlations in the measurements can be greatly enhanced.
We exploit this to demonstrate the reconstruction of mega-pixel sized fluorescent patterns behind the scattering tissue.
\end{abstract}

\setboolean{displaycopyright}{true}

\begin{document}

\maketitle

\section{Introduction}
Scattering of light is one of the main barriers preventing the imaging of fluorescent sources located deep inside biological tissue.
A microscope imaging a set of incoherent sources inside the tissue usually observes a noisy speckle pattern  that has little resemblance to the actual sources.


Despite the noise-like   appearance, speckle has strong statistical properties, such as the {\em memory effect} (ME), implying that the patterns generated by nearby sources are correlated.
It has been previously observed that due to this ME correlation, the auto-correlation of a speckle pattern generated by multiple independent sources is equivalent to the auto-correlation of the latent source layout~\cite{Katz2014,Bertolotti2012,Takasaki:14,Edrei:16,Edrei2016,Hofer:18, Wu:17,Wu:20,Wang:20,Chang2018}.
This fascinating property has drawn a lot of interest since it allows the recovery of latent illuminators,  completely invisible to the naked eye, purely by exploiting speckle statistics.


\begin{figure*}[h]
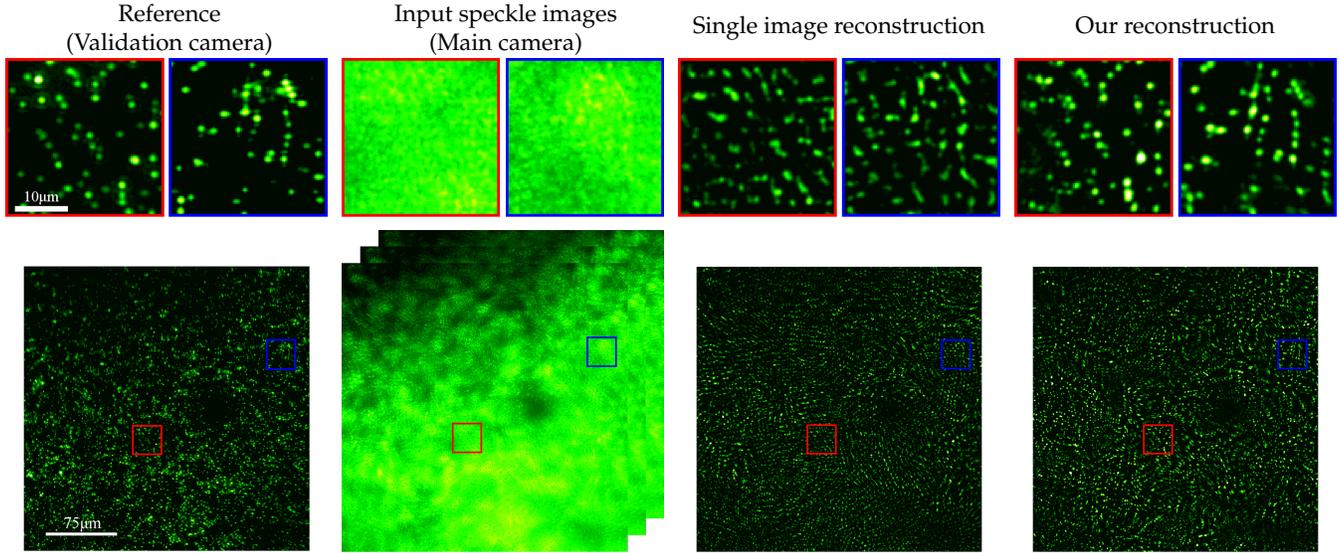

\centering
\begin{tabular}{*4{X}@{}}
\makecell{~Reference \\~(Validation camera)} & \makecell{~Input speckle images \\ ~(Main camera)} &Single image reconstruction & Our reconstruction\\
\comparepics{./figure/fluorescent/teaser_part} \\
\comparepics{./figure/fluorescent/teaser_marked}\\
\end{tabular}
\caption{\label{fig:teaser}\textbf{Reconstructing a wide-range fluorescent bead target from modulated speckles.} We reconstruct the layout of fluorescent beads spread behind a chicken breast tissue slice, whose thickness was measured at about $\sim150{\mu}m$. The beads are attached to the tissue,  separated only by a $150{\mu}m$ cover glass. The beads spread over a field of view of $300\mu m\times 300 \mu m$, occupying a one mega pixel image. While the ME correlations in a single image capture are too noisy to provide good reconstruction, we optically modulate the speckle field, capturing $54$ shots with different modulations. This modulation allows us to amplify statistical correlations, leading to accurate reconstruction of a complex illuminator pattern, despite high degradation and limited speckle contrast in the input images. The lower part of the figure includes the full images, while at the top we zoom on two sub windows for high resolution visualization.}
\end{figure*}

Despite its potential, there are still major challenges to solve before the idea can apply to realistic biomedical imaging scenarios.  The main barrier is that ME correlations are very weak and the amount of information that can be inferred from them is limited. To circumvent this, past experimental demonstrations have made use of various simplifying assumptions.
While fluorescent sources of interest in realistic biomedical imaging scenarios are located inside the scattering sample rather than far behind it,
most demonstrations of speckle correlation-based see-through algorithms consider sources located a few centimeters beyond the sample.
 This is due to the fact that when the sources are further from the scattering layer~\cite{SeeThroughSubmission}, they span a smaller range of angles relative to the sample and therefore speckle correlations are stronger.
A second issue is that the contrast of the observed speckle pattern decays as more independent emitters are present, and hence, the technique is mostly applicable  to very sparse emitter layouts.

In this work, we exploit strategies for maximizing the amount of ME correlation we can extract from speckle images.
To this end, we build on a simple observation: imaging the same layout of fluorescent sources through different tissue layers would allow us to obtain independent speckle images, thereby leading to independent auto-correlations.
Averaging such independent auto-correlations can suppress noise in the correlations and boost the quality of the illuminator patterns that we can infer from it.
While some temporal dynamics are present in live tissue, sequential images of illuminators inside the same tissue are still highly correlated. 
Rather, we use a programmable spatial light modulator (SLM) mask in the optical path imaging the tissue and use it to modulate the field, leading to different speckle images. We discuss various forms of modulation and arrive at a spatial form of modulation we term {\em \AlgName}. We show that despite the fact that we image the same tissue layer multiple times, we can get uncorrelated measurements that maximize ME correlation.

To further maximize the amount of information that we can extract from speckle data, we follow an idea recently proposed by~\cite{SeeThroughSubmission}, which argues that when light sources are located inside the sample rather than far behind it, the speckle pattern generated by each source has a limited support and does not spread over the full sensor. Thus, rather than computing a global full-frame auto-correlation, they compute local auto-correlations in the form of a Ptychography algorithm~\cite{zhou2019retrieval,Gardner:19,Li:2019:Ptycho,Li:B:2019:Ptycho,Shekel_2020}. These local correlations can boost the signal to noise ratio of the detected correlation by a few orders of magnitude.

Overall, we demonstrate the reconstruction of wide, complex fluorescent bead patterns inside scattering tissue. Our approach  captures only a few dozen images of the tissue, compared to hundreds of images used by recent approaches that image fluorescent sources behind scattering layers~\cite{Boniface2020,Zhu22} when using single-photon fluorescent emission; which is linear with respect to excitation energy.
Compared to recent wavefront shaping approaches~\cite{Boniface:19,Dror22} that only facilitate imaging of a local neighborhood governed by the limited extent of the ME, our approach recovers mega-pixel  images  over a wide field of view, as demonstrated in \figref{fig:teaser}.



\section{Principle} \label{sec:principle}
\subsection{A review of ME-based imaging}
We start with a quick review of the ME and its application for seeing inside scattering media.
Let $\oinp,\tinp$ denote the position of two illumination sources and $u^{\oinp}(\snsp),u^{\tinp}(\snsp)$ the fields they generate, where  $\snsp$ denotes a sensor coordinate.
The ME states that speckle fields generated by {\em nearby} sources are related by a tilt-shift correlation~\cite{osnabrugge2017generalized}.
Recently, Bar \etal \cite{single-sct-iccp-21} have offered a simple model for this relation, stating that
\BE\label{eq:me-corr-tilt-shift}
u^{\oinp}(\snsp)\approx e^{ik \alpha <\Dl,\btau>} u^{\tinp}(\snsp+\Dl),
\EE
with $\Dl=\tinp-\oinp$ the displacement between the sources,  $\btau=\snsp-\oinp$ the displacement between the source to the observation point, and $\alpha \approx \frac{-3}{2L}$, where $L$ is the tissue thickness.
This model assumes that we image the volume with a microscope whose sensor plane is conjugate to the plane of  the illuminators  $\oinp,\tinp$.

An image sensor only measures the intensity of the speckle pattern which we denote by
\BE\label{eq:S}
S^{\inp^n}(\snsp)= \abs{u^{\inp^n}(\snsp)}^2.
\EE
In the presence of multiple incoherent illuminators, we observe an intensity image $I(\snsp) = \sum_{n} S^{\inp^n}(\snsp) = \sum_{n} \abs{u^{\inp^n}(\snsp)}^2$.
Assuming source displacements are small enough for ME correlation to hold, speckle intensities from nearby sources are shifted versions of each other, $S^{\oinp}(\snsp)\approx S^{\tinp}(\snsp + \Dl)$. Note that since we deal with intensity images, phase adjustments are not required.

We denote by $S^{\bzero}(\snsp)$ the speckle from an illumination source at the center of the frame.
With this notation, we can express the sum of speckles from incoherent sources as
\BE\label{eq:conv}
I=\So(\snsp) \ast O,
\EE
where $O$ is a binary image denoting the location of the illumination sources, and $\ast$ denotes convolution. To detect fluorescent sources through scattering media, our goal is to recover the latent illuminator pattern $O$ from an input speckle image $I$.

We now filter $I$ and $\So$ to locally have a zero mean
\BE\label{eq:zero-mean}
\bI=I-g \ast I, \quad  \bSo = \So-g \ast \So,
\EE
where $g$ is a low pass filter.
We note that \equref{eq:conv} also holds if we replace $I,\So$ with $\bI,\bSo$, and we can express $\bI=\bSo \ast  O$.
Since $\bSo$ is a random zero mean signal, its auto-correlation is approximately an impulse function \cite{Bertolotti2012} \BE \label{eq:impulse-assump} \bSo \star \bSo \approx\delta.\EE   With this approximation \cite{Bertolotti2012,Katz2014}   derive the relationship:
\BE\label{eq:autocorrelation}
\bI \star \bI =  (\bSo \star \bSo) \ast (O\star O) \approx O \star O,
\EE
where $\star$ denotes cross correlation.
Thus, the auto-correlation of the input speckle intensity is equivalent to the auto-correlation of the desired latent image $O$. As a result, one can recover $O$ from $\bI \star \bI$ using a phase retrieval algorithm~\cite{Bertolotti2012}.

\paragraph{Challenges.} The observation made by \equref{eq:autocorrelation} is very compelling because it suggests that latent illuminators $O$ can be recovered from a noisy speckle image $I$, despite the fact that to the untrained eye, the input images carry no similarity to the latent source layout. Yet, it involves two major assumptions that limit
 its practical applicability.

 The first problem is that the ME correlation is not exact, especially when the displacement between the illuminators increase.
  A second problem is that since the speckle pattern $\So$  emerging from a single source  has a limited support and only spans a finite number of pixels, its auto-correlation outlined in \equref{eq:impulse-assump} is not a perfect impulse, but involves residual noise.  Effectively, in realistic scenarios, $\bI \star \bI$ is a {\em very} noisy approximation to $O \star O$. This noise increases  as more independent illuminators are present in $O$.

 In \figref{fig:autocorrelation}(b), we show the auto-correlation of a speckle image $\bI$ composed of a sparse layout of sources. We compare two layouts with a different number of sources. We can see that as more incoherent sources are included, the auto-correlation is very noisy and does not resemble $O\star O$. Our goal in this work is to improve the contrast of this auto-correlation by capturing multiple modulations of the speckle signal.

\newcommand*{\doublec}[1]{%
    \multicolumn{2}{c}{ \shifttext{-2em}{#1}}
}

\begin{figure}[t]
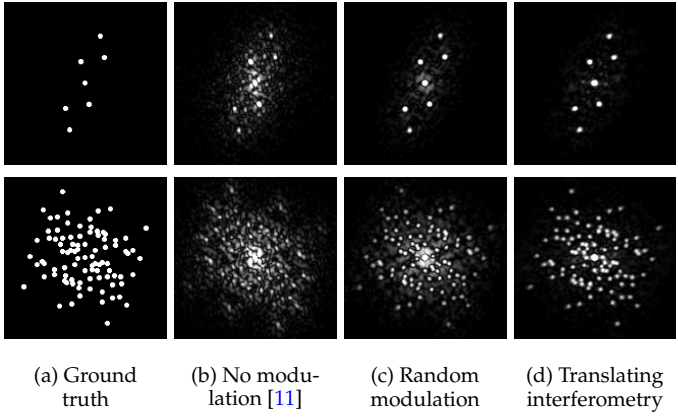

\centering
\begin{tabular}{*4{P}@{}}
\comparecorr{./figure/compare/collisionfree_fewer_A} \\
\comparecorr{./figure/compare/collisionfree_A} \\
\small{(a) Ground truth} & \small{\makecell{(b) No modu-\\lation \cite{SeeThroughSubmission}}}& \small{(c) Random modulation}& \small{(d) Translating interferometry}
\end{tabular}
\vspace{-1em}
\caption{\label{fig:autocorrelation}\textbf{Comparing speckle auto-correlation with different modulation approaches.} The two rows compare illuminator layouts with different complexities, while the columns evaluate different modulation strategies. All results use the same number of shot images. (a) Ground truth  auto-correlation. (b) Without any modulation, reconstructed auto-correlation is noisy, and when the target is complex (2nd row), it is almost unrecognizable. (c) Random modulation can improve contrast, but still contains noise. (d) {\AlgNameC} can clearly recover the auto-correlation. }
\end{figure}

\subsection{Improving auto-correlation contrast} \label{sec:impcorr}
To analyze the contrast of the speckle correlation, we introduce the following notation. We denote by $\Listdl$ the set of all displacements $\Dl$, such that our latent pattern includes a pair of illuminators $(n,m)$ displaced by $\Dl$ :
\BE
\Listdl=\left\{\Dl|\exists (n,m),\;\; \Dl=\inp^m-\inp^n \right\},
\EE
and by $\CListdl$ the list of all other displacements. We denote the speckle auto-correlation by $\bC^{\bI}$, which is defined for a displacement $\Dl$ as:
\BE\label{eq:auto-corr}
\bC^{\bI}(\Dl)=\sum_\snsp \bI(\snsp)\bI(\snsp+\Dl).
\EE
Intuitively, the speckle correlation has good contrast if $\bC^{\bI}(\Dl)$ is high for  displacements  $\Dl\in \Listdl$, corresponding to real illuminator positions; and is low for all other displacements $\Dl\in \CListdl$.
We define the correlation contrast using the following signal to noise metric:
\BE\label{eq:contrast}
	\CorCont\left(\bC^{{\bI}}\right)=\frac{\frac{1}{\abs{\Listdl}}\sum_{\Dl\in \Listdl}\E{\bC^{{\bI}}(\Dl)}^2}{\frac{1}{\abs{\CListdl}} \sum_{\Dl\in \CListdl}\Var{\left[ \bC^{{\bI}}(\Dl)\right]}}
\EE
One way to increase the correlation contrast used by \cite{Katz2014}, is to capture multiple images of the latent pattern $O$ behind different scattering layers. Effectively, we measure $I_t=S^\bzero_t\star O$  with different speckle patterns $S^\bzero_t$. The auto-correlation is then evaluated as the average of the individual auto-correlations
\BE\label{eq:avg-auto-corr}
\bbC(\Dl)=\avgt \bC^{\bI_t}(\Dl),
\EE
with $\bC^{\bI_t}=\bI_t\star \bI_t$ as defined in \equref{eq:auto-corr}.

In  supplement Sec. A.A, we formally prove
the following.
\begin{claim}\label{claim:cor-cont-linear}
If the speckle patterns $S^\bzero_t$ are uncorrelated with each other for different $t$ values, then replacing $\bC^{\bI}$ with $\bC^{\bI_1,\ldots,\bI_T}$ in the correlation contrast of \equref{eq:contrast} increases the contrast {\em linearly} with the number of measurements $T$, i.e.,
\BE
\CorCont \left(\bbC\right)=T\cdot \CorCont \left(\bC^{\bI}\right).
\EE
\end{claim}

While this is a promising idea, when the sources are located inside the tissue, it is not easy to image the same illuminators through different scattering layers.
Rather, in this work, we would like to modify the speckle patterns by adjusting the optics.

\paragraph{Random modulation. } Intuitively, to create different speckle intensity images, we can put a random phase mask in the optical path between the sample to the imaging sensor. If we put this mask in the Fourier plane, it would translate into a convolution of the fields $u^{\inp}(\snsp)$ with the Fourier transform of the mask, which we denote as $h_t$.
This would lead into an intensity image $I_t=\sum_n S^{\inp^n}_t$ with
\BE\label{eq:randS}
S^{\inp^n}_t=\abs{u^{\inp^n} \ast h_t}^2.
\EE
In \figref{fig:autocorrelation}(c), we compare the auto-correlation of a single speckle image to the average auto-correlation with 54 random masks $h_t$. Averaging random masks rejects noise and improves the correlation contrast, but it is still noisy.

To understand why random modulation is sub-optimal, we review the tilt-shift correlation in \equref{eq:me-corr-tilt-shift}. If the fields $u^{\oinp},u^{\tinp}$ generated by different illuminators  would follow a pure shift,
 then $u^{\oinp} \ast h_t, u^{\tinp} \ast h_t$ would also be shifted versions of each other. However, according to \equref{eq:me-corr-tilt-shift}, fields from different sources vary by phase, and hence a convolution with $h_t$ largely degrades the correlation and  $S^{\oinp}_t(\snsp)$ would differ from $S^{\tinp}_t(\snsp + \Dl)$.

\paragraph{Translating interferometry. } Our goal in this work is to change the optical path such that we can capture multiple uncorrelated  speckle patterns, and yet  maintain the ME correlation.
To this end, we design an interferometric setup that allows us to measure the interference between $u^{\inp^n}$ and a shifted copy of it, which we name {\em \AlgName}. This leads  to a measurement of the form
\BE\label{eq:ourS}
S^{\inp^n}_t= u^{\inp^n}(\snsp)u^{\inp^n}(\snsp+ \pt)^*,
\EE
where $\pt$ denotes the displacement vector. We acquire these measurements using an incoherent interferometry scheme, described in supplement Sec. A.C.

When several incoherent sources are present, we will acquire an incoherent summation
\BE \label{eq:intef-It} I_t=\sum_n S^{\inp^n}_t.\EE This interferometric measurement is already a zero mean signal and there is no need to subtract the mean as with the intensity measurements of \equref{eq:zero-mean}.

 The \AlgName measurements provide two main benefits which we summarize in the following claims and prove in supplement Sec. A.A. First, unlike a naive optical mask in \equref{eq:randS}, it does not reduce the ME correlation of the original speckles. Second, despite the fact that these measurements are captured from the same tissue layer and they are not independent, they are still {\em uncorrelated}.
\begin{claim}
For displacements in the order of a few speckle grains, the correlation between \AlgName signals $S^{\oinp}_{t},S^{\tinp}_{t}$ produced by different illuminators $\oinp,\tinp$ is approximately the same as the correlation of the original speckle intensity images.	
\end{claim}

\begin{claim}\label{clm:uncorrelated}
For displacements ${\opt},{\tpt} $ whose distance $\|{\opt}-{\tpt}\|$ is larger than the speckle grain, the signals $S^{\inp^n}_{t_1},S^{\inp^n}_{t_2}$ are uncorrelated.
\end{claim}
The observation in \clmref{clm:uncorrelated} is central for this paper.
The fact that different displacements lead to uncorrelated speckle measurements  means that according to \clmref{claim:cor-cont-linear}, we could average them and the auto-correlation contrast would improve {\em linearly} with the number of measurements.

The auto-correlation of the \AlgName measurements relates to the auto-correlation of the hidden illuminator pattern $O$, but unlike pure intensity speckles, with the above modulations a phase correction is needed, which we derive in the following claim, and prove in supplement Sec. A.A.

\begin{claim}\label{claim:trans-interf-auto-corr}
	Using the \AlgName measurements of \equref{eq:ourS}, the speckle auto-correlation  $\bC^{I_t}=I_t\star I_t$ approximates  the auto-correlation of the latent pattern $\bC^O=O\star O$,  times a phase ramp correction
	\BE\label{eq:phase-ramp-corr}
	\bC^{I_t}(\Dl)\approx e^{-jk\alpha<\Dl,\pt>}\bC^{O}(\Dl).
	\EE
	\end{claim}

Given the relation in \clmref{claim:trans-interf-auto-corr}, we average the auto-correlation of the different \AlgName measurements, applying the  phase ramp correction of \equref{eq:phase-ramp-corr}:
\BE\label{eq:avg-auto-corr-ramp}
\bnbC=\avgt e^{jk\alpha<\Dl,\pt>} \bC^{I_t}(\Dl)
\EE
The phase corrected averaging in \equref{eq:avg-auto-corr-ramp} is subject to a single unknown parameter $\alpha$. In our implementation, we manually tune it to maximize the visual quality of the results.

\paragraph{Analyzing \AlgName  correlations.}
In \figref{fig:autocorrelation}(d) we show the auto-correlation obtained by averaging  \AlgName measurements $I_t$ (\equref{eq:intef-It}) with the phase  ramp  correction of \equref{eq:avg-auto-corr-ramp}.  Our approach reduces noise and improves the correlation contrast when compared with random modulations (\equref{eq:randS}) or just with the auto-correlation of a single speckle image.

In \figref{fig:autocorrelation}(d), we average $18$ \AlgName measurements $I_t$.  Note that as we explain in supplement Sec. A.C, 
capturing each interferometric measurement $I_t$  requires $3$ shots, so the $18$ measurements in \figref{fig:autocorrelation}(d-e) are acquired using a total of $54$ shots. This is compared against $54$ independent measurements captured by the random modulation approach. The $18$ \AlgName  modulations are  superior over the $54$ random modulations.

\begin{figure}[t]
     \centering
     \begin{subfigure}[b]{0.22\textwidth }
         \centering
         \includegraphics[width = \textwidth]{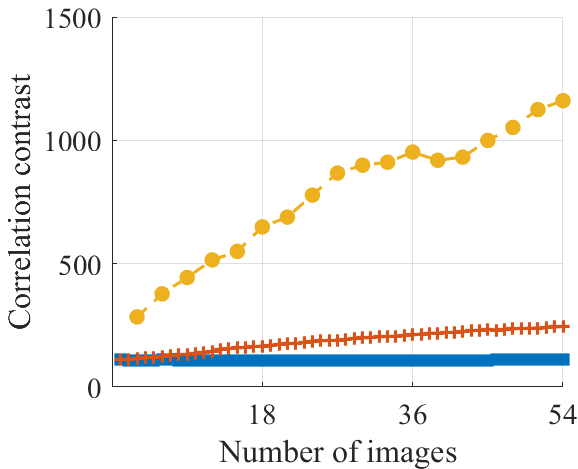}
     \end{subfigure}
     \begin{subfigure}[b]{0.22\textwidth}
         \centering
         \includegraphics[width = \textwidth]{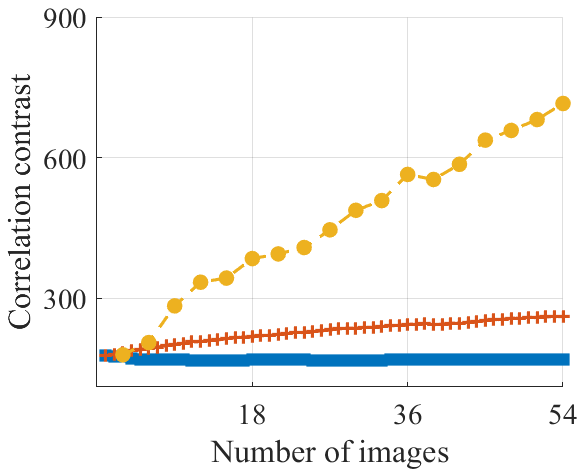}
     \end{subfigure}
    \newline
     \begin{subfigure}[b]{0.45\textwidth}
         \centering
         \includegraphics[width = \textwidth]{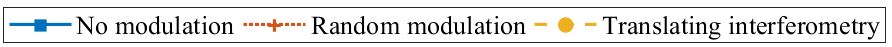}\\
     \end{subfigure}
\caption{\textbf{Contrast improvement for different modulations.} We compare speckle correlation contrast as a function of the number of images for different modulation approaches. The two images show correlation contrast measurements for two different tissue slices. The graphs are noisy due to the finite speckle spread in an image, but still demonstrate clear differences between different modulation approaches. Without modulation, multiple images only reduce read and photon noise of the image, which does  not lead to a significant improvement in contrast. Multiple images with random modulations can improve contrast, but the gain saturates quickly as ME degrades. 
On the other hand, \AlgName can achieve a higher contrast, and, in agreement with theory, contrast increases roughly linearly with the number of shots.
}
\label{fig:compScontrast}
\end{figure}

\begin{table}[t]
\small
\caption{\label{tab:compSsumm} \textbf{Summary of different speckle modulation approaches.}}
\centering
\begin{tabular}{cc@{\hskip -0.1em}c}
\toprule
                       & $S^{\inp^n}_t$                                   & \small{Averaged auto-correlation}\\
No modulation \cite{SeeThroughSubmission}          & $\abs{u^{\inp^n}(\snsp)}^2$                      & $ \bC^{\bI_t}(\Dl) $\\
Random modulation  & $\abs{u^{\inp^n}(\snsp) \ast h_t}^2$             & $\avgt \bC^{\bI_t}(\Dl) $\\
Translating interf.  & $u^{\inp^n}(\snsp)u^{\inp^n}(\snsp+ \pt)^*$        & $\avgt e^{jk\alpha<\Dl,\pt>} \bC^{I_t}(\Dl)$\\
\bottomrule
\end{tabular}
\end{table}

We note that the \AlgName measurements used here are similar to
those used in shearing interferometry \cite{riley1977laser}.
However, shearing interferometry usually uses smaller displacements to obtain the local gradient of the wave, while the displacements we use here are larger than the speckle grain size so that we obtain uncorrelated speckles.

Table \ref{tab:compSsumm} summarizes the different modulation approaches evaluated in this paper.
In \figref{fig:compScontrast} we plot the correlation contrast of \equref{eq:contrast} as a function of the number of averaged images $T$.
 We start by capturing multiple images without any modulation. This only reduces  read and photon noise, which  does not translate into a real improvement in correlation contrast. When we randomly modulate the wave (\equref{eq:randS}), the contrast increases but it eventually saturates as the convolution reduces the ME correlation.  By using our {\AlgName} (\equref{eq:ourS}), the contrast increases roughly linearly as predicted by the theory.
 This suggests that the speckle signals $S_t$ we generate are indeed uncorrelated for different displacements $\opt,\tpt$.
The graphs in  \figref{fig:compScontrast} demonstrate correlation contrast observed with two different tissue slices. For each of the tissue slices, we generated  the source layout in the top row of   \figref{fig:autocorrelation}.
As we evaluate speckles through real tissue, we note that: 1) the exact amount of correlation we measure in each tissue slice can vary, and 2) each tissue layer generates a speckle spread with a limited support. As we only average a finite number of speckle pixels, the graphs are noisy.
Despite these issues,  the graphs measured from difference slices demonstrate consistent trends.

\subsection{Exploiting local support}\label{sec:local-sup}
The previous section aims to increase the  auto-correlation contrast by averaging multiple measurements.
To improve on it, we adopt a complementary approach for noise reduction, recently suggested by \cite{SeeThroughSubmission}. This approach is based on the observation that when the light sources are inside the sample, rather than far behind it, the speckle pattern scattered from a single source has {\em local support}, \textit{i.e.}, the scattered light does not spread over the entire sensor.
Therefore, it is argued that computing the full-frame auto-correlation over the entire image corrupts the signal with additional noise. Rather,  it is sufficient to match between the local correlations of the observed speckle pattern and the optimized latent image. This leads into a Ptychography style cost~\cite{PhysRevLett.98.034801}.
We review the exact cost in supplement Sec. A.B.
In the experimental section below, we show that moving from full-frame correlations to local ones has a major impact on noise elimination and improving the resulting reconstruction.

Another advantage of the local cost discussed in \cite{SeeThroughSubmission}, is that it allows recovering patterns larger than the extent of the ME. As mentioned above, ME correlations of the form of \equref{eq:me-corr-tilt-shift} only hold for small displacements $\Dl$. When matching the full-frame auto-correlation (\equref{eq:autocorrelation}) of $I$ and $O$, we rely on the fact that ME correlation exists between any two sources in our latent pattern. This assumption largely limits the  range of recoverable illuminator patterns to patterns lying within the ME range. In contrast,
the local cost only relies on local correlations between sources  in the same local window.
At the same time, the overall extent of the illuminator pattern $O$ can be larger than these local windows.

\section{Results}
\label{sec:exp}
\subsection{Experiment setup}
\figref{fig:setup} illustrates our acquisition setup, including an imaging arm and a validation arm.
The imaging arm consists of an objective and a tube lens, followed by a second relay system which allows us to place a spatial light modulator (SLM) at the Fourier plane. The image of the modulated field is collected by the main camera.
The objective attempts to image incoherent sources beyond a scattering sample.
A second validation camera images the beads from the other end of the tissue, allowing the capture of a clear unscattered  image of the illuminator layout, which is used  to assess reconstruction quality. Note that this validation camera does not provide any input to the algorithm.
The target and the validation objectives are mounted on z-axis translation stages, facilitating accurate  control over focusing in both imaging and validation arms.

\begin{figure}[t]
\centering
\includegraphics[width=25em]{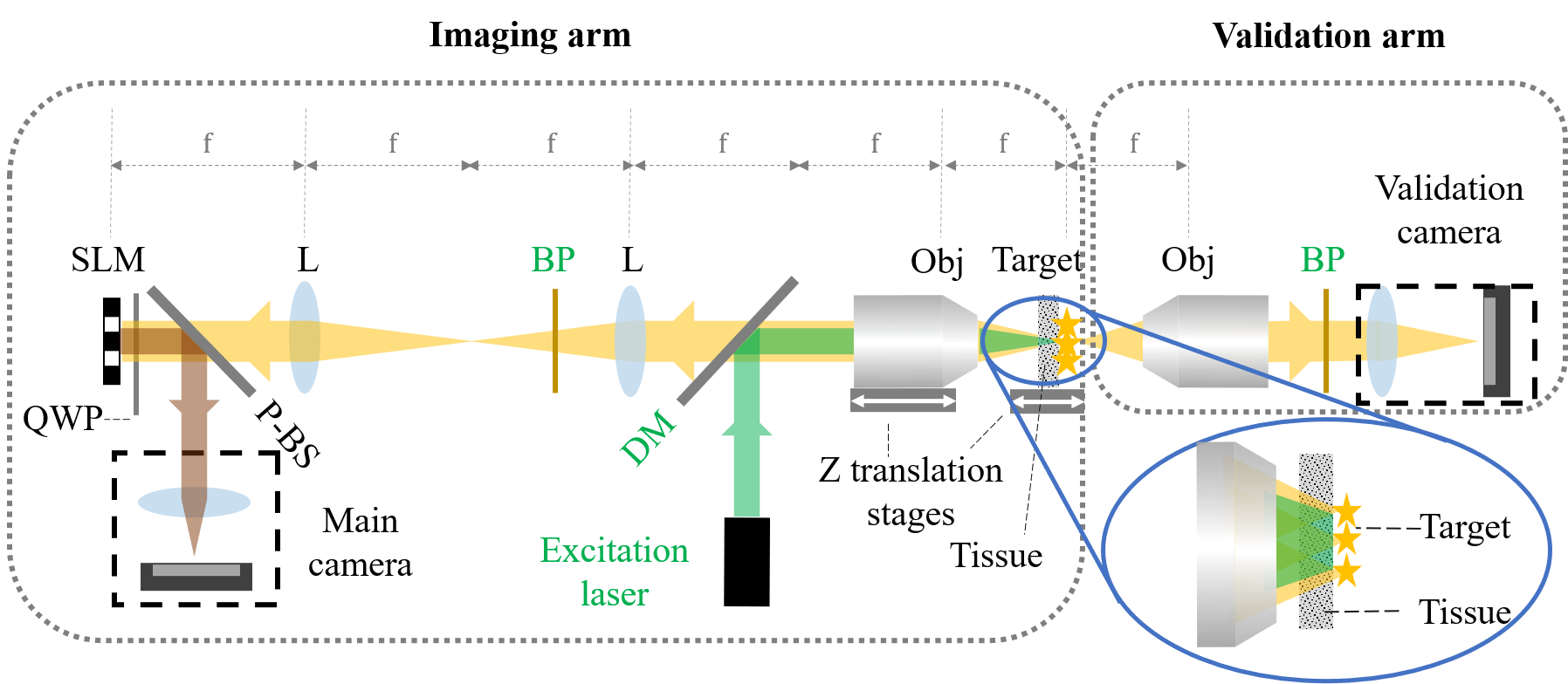}
\caption{\label{fig:setup}\textbf{Experimental setup.} L: lens. Obj: objective lens. BP: bandpass filter. QWP: quarter wave plate. P-BS: polarized beam splitter. DM: dichromic mirror. 
		We attach fluorescent beads at the back of a tissue layer. A laser excites the beads from the  front side of the tissue. The emitted light is back scattered through the tissue and collected by the main camera. We create a 4f relay system in the optical path and place an SLM in its Fourier plane to modulate the scattered light.
		Finally, we used a validation camera behind the tissue which can image the beads directly. This camera is not part of our algorithm and is only used to validate its output.}
\end{figure} 

For most of our experiments, we used chicken breast tissue as a scattering sample. In the supplement, we also demonstrate results imaging through a parafilm tissue phantom. We discuss what is known about the optical characterization of these materials in supplement Sec. A.H.
We  used Spherotech Fluorescent Nile Red Particles $0.4-0.6\mu m$, FP-0556-2.  The beads are attached on a microscope cover glass behind the scattering tissue. The separation between the beads and the tissue is as low as $150\mu m$, the thickness of the cover glass. The beads are excited with a 530nm laser from the front side of the tissue. The excitation light scatters through the tissue, illuminates the beads, and the emitted light scatters back through the tissue to the camera. We filtered the excited light using a 10nm bandpass filter centered at $580$nm.


Throughout the experimental section, we visualize the fluorescent bead captures with a green colormap.
We use a red colormap to visualize images from an alternative setup described in supplement Sec. A.D. Such images are captured through real tissue but the source layout  is computer-synthesized.

\begin{figure*}[t!]
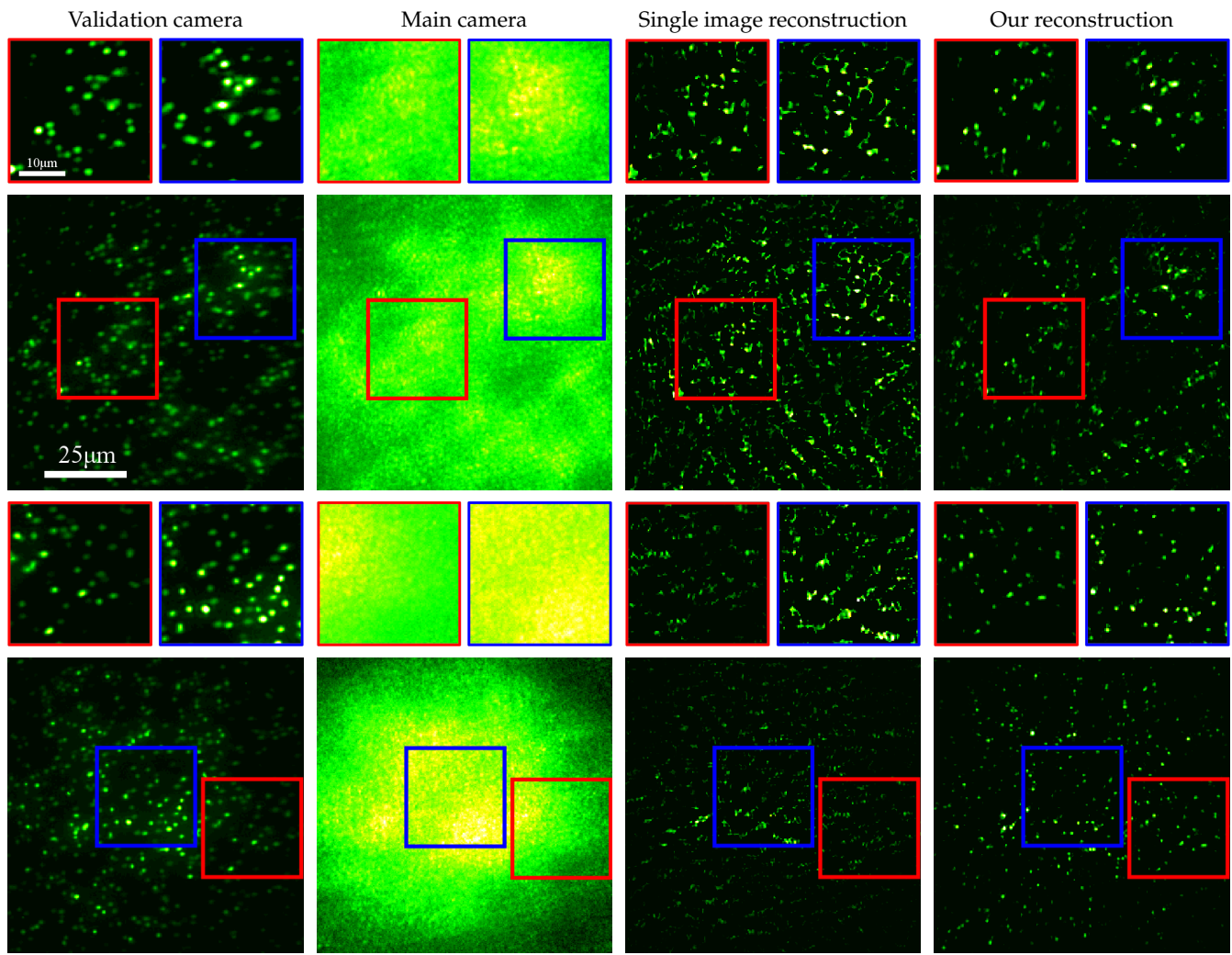

\centering
\begin{tabular}{*4{X}@{}}
Validation camera & Main camera & Single image reconstruction & Our reconstruction\\
\comparepics{./figure/fluorescent/150um_dense2_redo_part} \\
\comparepics{./figure/fluorescent/150um_dense2_redo_marked}\\
\comparepics{./figure/fluorescent/150um_dense_redo_part} \\
\comparepics{./figure/fluorescent/150um_dense_redo_marked}
\end{tabular}
\caption{\textbf{Reconstruction results for fluorescent beads behind \textasciitilde$150\mu m$-thick tissue slices.} With \AlgName, we can clearly reconstruct fluorescent bead targets from $54$ shot images captured by the main camera. The reconstruction is compared against a reference image from the validation camera observing the beads directly. In contrast, a standard single-image shot of the scattered light only facilitates a very noisy reconstruction. }
\label{fig:showfluorescent}
\end{figure*}

\begin{figure*}[h!]
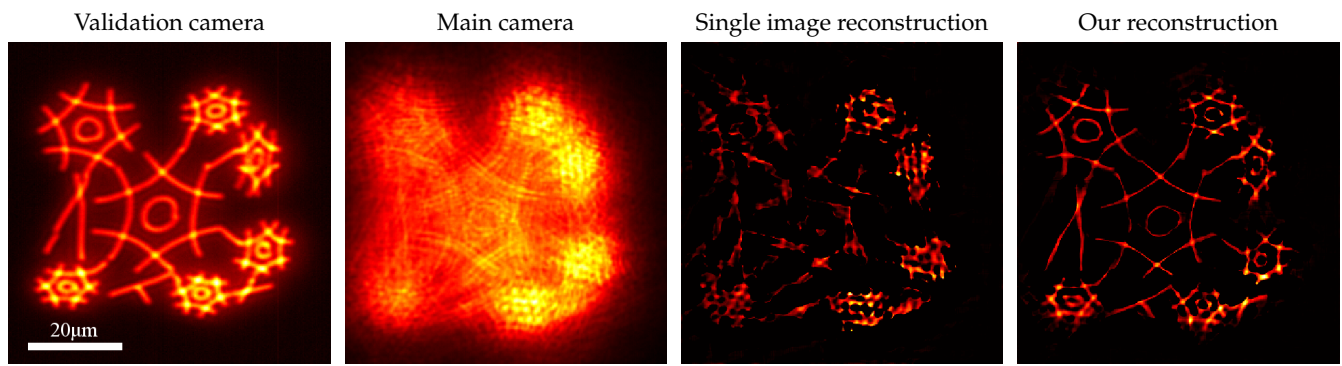

\centering
\begin{tabular}{*4{X}@{}}
Validation camera & Main camera & Single image reconstruction & Our reconstruction\\
\comparepics{./figure/nearfield/150um/neural_network}
\end{tabular}
\caption{\textbf{Reconstruction results from a spatially incoherent target.} We reconstruct a spatially incoherent target with a more complicated layout behind \textasciitilde$150\mu m$-thick tissue slice.}
\label{fig:showlaser}
\end{figure*}

%
%


\subsection{Experiment results}
We start by demonstrating our setup on a fluorescent bead target.
 In \figref{fig:teaser}, we reconstruct a $1000\times 1000$ pixel image, corresponding to a field of view of $300{\mu}m\times 300{\mu}m $. The random beads were spread behind a \textasciitilde$150{\mu}m$ thick tissue. Note that all thickness measurements in this paper  are approximated due to the limited resolution of the  clipper. The bead layout is unrecognizable from the captured speckle input. Moreover,  as so many independent incoherent sources are present, the input images are rather smooth and speckle variation is almost invisible.
  Despite this, our {\AlgName} framework achieves a clear reconstruction from $54$ shots.
 \figref{fig:showfluorescent} demonstrates another reconstruction of beads behind a \textasciitilde$150{\mu}m$ thick tissue.
\figref{fig:showlaser} demonstrates a reconstruction of structured patterns  rather than  sparse beads. Details about this target are provided in supplement Sec. A.D.

 While the reference and the reconstruction have the same layout they have  somewhat different brightness and resolution. The resolution of the reference is subject to the diffraction limit. The reconstruction algorithm on the other hand encourages sparse results, and hence, recovered dots tend to be narrower.
 The brightness variation is partially attributed to the fact that the reference is captured by a different camera from a different direction, but also due to imperfect convergence of our optimization.



\begin{figure*}[!t]
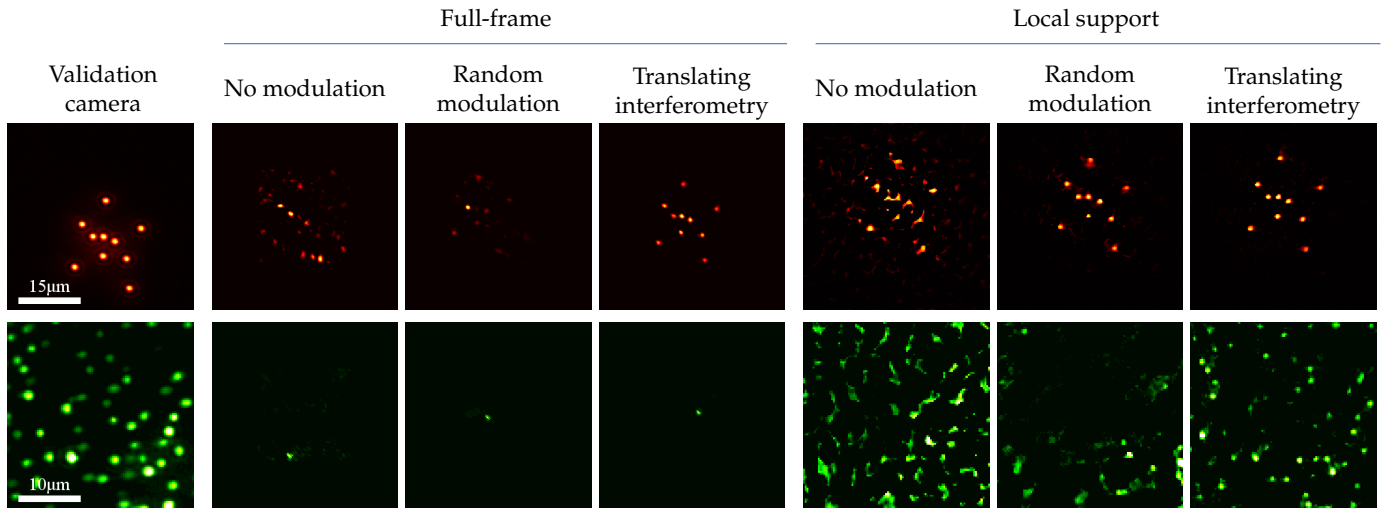

\centering
\begin{tabular}{S@{\hskip 1.8em}SSS@{\hskip 1.8em}SSS}
 \multirow{2}[13]{*}{ \makecell{Validation\\camera}}&  \multicolumn{3}{c}{ \shifttext{-2em}{Full-frame}}  & \multicolumn{3}{c}{ \shifttext{-1em}{Local support}} \\ \cmidrule(r{1.2em}){2-4} \cmidrule{5-7}
 &  No modulation & Random modulation & \makecell{Translating\\ interferometry}  & No modulation & Random modulation & \makecell{Translating\\ interferometry}\\
\comparealg{./figure/compare/collisionfree_ffrecon}{./figure/compare/collisionfree_recon} \\
\comparealg{./figure/compare/fluorescent_ffrecon_redo}{./figure/compare/fluorescent_recon_redo}
\end{tabular}
\caption{\textbf{Comparing  reconstruction and modulation approaches.} Top part: a spatially incoherent target with a sparse and simple layout. Lower part: challenging fluorescent beads target. Local support correlations~\cite{SeeThroughSubmission} are stronger than standard full-frame auto-correlations~\cite{Chang2018}, and our {\AlgName} improves over a single shot (with no modulation) and over simple random modulation. For the simple target on the top, the full-frame algorithm can recover the image given the improved correlation provided by {\AlgName} modulations, but fails to do so from the noisier correlations provided by other modulation strategies. The local correlation approach which is more robust to noise can recover the target even with the simpler modulations.  For the challenging target at the lower part, the full-frame algorithm fails completely using all types of modulations. The local correlations algorithm can reconstruct the target using   {\AlgName} modulations. However given random modulations, it can only reconstruct a subset of the beads.}
\label{fig:compalg}
\end{figure*}

\subsection{Comparing reconstruction and modulation approaches}
In \figref{fig:compalg}, we evaluate two components of our algorithm: (i) the usage of local correlations~\cite{SeeThroughSubmission} discussed in \secref{sec:local-sup}, versus the standard full frame auto-correlation used in previous work~\cite{Katz2014,Bertolotti2012,Chang2018}, and
 (ii) the modulation approach. 

 As discussed in~\cite{SeeThroughSubmission} and reviewed in \secref{sec:local-sup}, the local approach detects correlations with a higher SNR compared to the full frame approach and indeed it leads to better reconstructions.
 We also show that the \AlgName modulation leads to better results compared to simpler modulation alternatives.

 \figref{fig:compalg} compares different modulation schemes, as well as evaluates the effectiveness of full-frame and local-correlation algorithms.
 The top row of the figure shows reconstructions of a spatially-incoherent target with a simple structure.
For the full frame approach, a single shot results in  unrecognizable reconstruction, that is only slightly improved given $54$ random modulations.
In contrast, our {\AlgName}  can correctly reconstruct the pattern.
The usage of local correlations compared to full-frame ones reduces some of the noise, and hence  even  random modulations can lead to good reconstruction.

The bottom row of the figure compares reconstructions on  a  denser  fluorescent bead target. As explained in~\cite{SeeThroughSubmission}, the increased source density is more challenging to reconstruct as the contrast of the incoherent speckle image decreases.
The full-frame approach fails to reconstruct this target with any of the modulation approaches.
The local correlation algorithm fails with a single shot (no modulation). The random modulation reconstructed only a subset of the beads, and the best results are obtained using {\AlgName} modulations.

\subsection{Evaluating reconstruction vs. number of images}
\figref{fig:compScontrast} numerically evaluates the correlation contrast improvement as a function of  the number of images.
In  \figref{fig:compimgnumrecon}, we visually compare reconstructions using an increasing number of input images, and demonstrate how the improved contrast translates into better reconstruction quality.
This experiments uses the same bead layout as in the lower part of  \figref{fig:compalg}, please refer to the validation camera reference displayed there.

%
%

\begin{figure}[t]
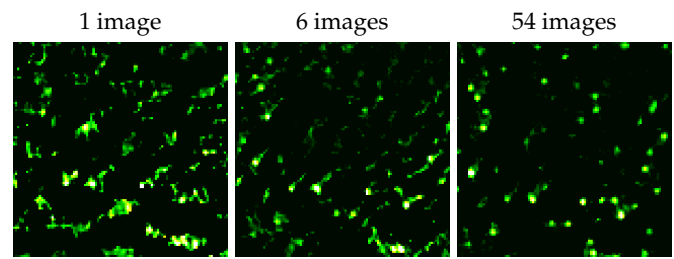

\centering
\hspace{-1em}
\begin{tabular}{*3{M}@{}}
1 image & 6 images & 54 images\\
\compareimagenum{1} & \compareimagenum{6} & \compareimagenum{54}
\end{tabular}
\caption{\textbf{Evaluating reconstruction vs.  number of shots}. We visualize the quality of the reconstruction given no modulation (single shot), and with an increasing number of modulation patterns. As evaluated numerically in \figref{fig:compScontrast} the correlation contrast increases linearly with the number of shots and in agreement, the visual quality of the reconstruction improves.   The reference speckle layout in this experiment is equivalent to the  validation camera image in the lower part of \figref{fig:compalg}.}
\label{fig:compimgnumrecon}
\end{figure}




\section{Discussion}

In this work, we demonstrate the reconstruction of fluorescent illuminator patterns attached to a scattering tissue. Despite heavy scattering that completely distort the captured images, we can exploit ME correlations in the measured speckles to detect the hidden illuminator layout.
While such correlations are inherently weak, we suggest a modulation scheme which allows us to capture multiple uncorrelated speckle measurements of the same sources. By averaging the correlations of such measurements, we increase the signal to noise ratio in the data and largely boost the reconstruction quality. We combine these modulated measurements with a recent algorithm~\cite{SeeThroughSubmission} seeking a latent pattern whose local correlations agree with local correlations in the measured data. The local correlations provide additional improvement in SNR. Moreover, since it only assumes local correlations between the speckles emitted by nearby sources rather than global full-frame correlations between any two sources in the image, it allows us to reconstruct wide patterns, much beyond the limited extent of the ME.
Overall, we demonstrate the reconstruction of mega-pixel wide patterns, limited only by the sensor size.

Despite the advances offered by this approach, it is inherently dependent on the existence of some ME correlation; thus when such correlations are too weak to be measured, our approach will fail as well.
Two major factors that decrease the ME correlations are the tissue thickness and the number of independent incoherent sources. In our current implementation, we recovered sources behind $150{\mu}m$-thick tissue, which is beyond the penetration depth of a standard microscope. While this significantly advances the capabilities of a standard microscope, numerous biomedical applications would benefit from increasing this depth further.
%
The second challenge for ME-based correlations is that we assume speckle variation is observed in the captured data. As more and more independent sources are present, the speckle contrast decays and the captured intensity images are smoother. Naturally, when speckle contrast is lower than the photon noise in the data, no correlations can be detected. Hence ME-based techniques are inherently limited to simple, sparse illuminator layouts.

An alternative approach for seeing through scattering tissue is based on wavefront shaping optics. Rather than post-process the speckle data, it attempts to modulate the incoming excitation light and/or the outgoing emission, to undo the tissue aberration in optics. In theory, this approach carries the potential to extend to thicker tissue layers and to correct complex patterns. In practice, efficiently finding a proper modulation mask is a challenging task. Recently, \cite{Dror22} has managed to recover such a modulation efficiently using linear (single photon) fluorescent feedback from a sparse set of beads. However, even after recovering a good modulation mask, the area they could correct with it is limited due to the limited extent of the ME. In contrast, our approach can recover wide-field-of-view, full-frame images, as it only relies on local ME correlations.

\begin{backmatter}
\bmsection{Disclosures} The authors declare no conflicts of interest.

\bmsection{Supplemental document}
See Supplement 1 for supporting content.

\end{backmatter}

\bibliography{biblio_anat,references}

\bibliographyfullrefs{references}

\ifthenelse{\equal{\journalref}{aop}}{%
\section*{Author Biographies}
\begingroup
\setlength\intextsep{0pt}
\begin{minipage}[t][6.3cm][t]{1.0\textwidth} 
  \begin{wrapfigure}{L}{0.25\textwidth}
    \includegraphics[width=0.25\textwidth]{john_smith.eps}
  \end{wrapfigure}
  \noindent
  {\bfseries John Smith} received his BSc (Mathematics) in 2000 from The University of Maryland. His research interests include lasers and optics.
\end{minipage}
\begin{minipage}{1.0\textwidth}
  \begin{wrapfigure}{L}{0.25\textwidth}
    \includegraphics[width=0.25\textwidth]{alice_smith.eps}
  \end{wrapfigure}
  \noindent
  {\bfseries Alice Smith} also received her BSc (Mathematics) in 2000 from The University of Maryland. Her research interests also include lasers and optics.
\end{minipage}
\endgroup
}{}
\clearpage
\appendix
\pagenumbering{arabic}

\begin{strip}
	\noindent{\huge {\bf Enhancing Speckle Statistics for Imaging Inside
Scattering Media}}\\\\
	\noindent{\Large {\bf Supplementary Appendix}}
\end{strip}

\renewcommand{\theequation}{S\arabic{equation}}
\renewcommand{\thefigure}{S\arabic{figure}}    

\setcounter{equation}{0}
\setcounter{figure}{0}

\section{Appendix}
\subsection{Analyzing correlation properties in \AlgName measurements}
In this section, we formulate and prove a few properties of \AlgName modulations mentioned in the main paper.

Throughout this derivation, we assume that different pixels $\snsp$ in a speckle pattern $u^{\inp^n}(\snsp)$ arising from one source are independent random variables. In practice, depending on magnification, the  grain of speckle features on the sensor may be wider than a single pixel. This would imply some scaling adjustments in the exact formulas, which we neglect here.
We also ignore small dependencies introduced by the low pass filtering of \equref{eq:zero-mean} and treat entries of $\bS^{\inp^n}(\snsp)$ as independent random variables.

Also, in the following derivations expectations
are taken with respect to multiple realization of scattering volumes with the same material parameters.
For example, the notation  $\E{\bC^{{I_1}}(\Dl) }$ refers to the following. Suppose we had $N$ different tissue layers of the same type, and we place behind each layer fluorescent sources at the exact same layout.
We image $N $ different speckle images and compute $N$ different auto-correlations  $\bC^{{I_1}}(\Dl) $. For large $N$ values averaging these auto-correlations provides the idealized expected speckle-auto correlation for that source layout, denoted $\E{\bC^{{I_1}}(\Dl) }$.
Similarly $\Var{\left[\bC^{{I_1}}(\Dl)\right] }$ denotes the variance  we expect to see in such auto-correlation. As correlation is computed from a finite number of speckle pixels the correlation is never zero even in displacements $\Dl$ that do not correspond to an actual illuminator displacement.
To define variance mathematically one needs to compute such auto-correlations from  $N$ different speckle images of the same illuminator layout. In practice for the evaluation in \figref{fig:compScontrast} we only compute the expectation and variance between different $\Dl$ displacements of the auto-correlation of a single tissue sample.

\setcounter{claim}{0}

\begin{claim}\label{claim:cor-cont-linear}
	We define correlation contrast
	\BE\label{eq:contrast-app}
	\CorCont\left(\bnbC\right)=\frac{\frac{1}{\abs{\Listdl}}\sum_{\Dl\in \Listdl}\E{\bnbC(\Dl) }^2}{\frac{1}{\abs{\CListdl}} \sum_{\Dl\in \CListdl}\Var{\left[\bnbC(\Dl)\right]}}
	\EE
	
	If the speckle patterns $S^\bzero_t$ are uncorrelated with each other for different $t$ values, than the correlation contrast increases {\em linearly} with the number of measurements $T$.
\end{claim}
\begin{proof}
	The claim is based on the observation that for displacements $\Dl\in \Listdl$ that correspond to an actual illuminator's displacement, $\E{\bnbC(\Dl)}$ is a positive quantity, while for  $\Dl\in \CListdl$, which do not correspond to a displacement between two illuminators, no correlation exists and in expectation $\E{\bnbC(\Dl)}=0$.
	
	With this understanding, we note that expectation is linear and hence recalling the definition of $\bnbC(\Dl) $ in \equref{eq:avg-auto-corr} of the main paper, the numerator of \equref{eq:contrast-app} is independent of the number of measurements $T$:
	\BE
	\E{\bnbC(\Dl)}=\frac{1}{T}\sum_t \E{\bC^{{I_t}}(\Dl)}=\E{\bC^{{I_1}}(\Dl)}.
	\EE
	
	We now move to express the denominator. First, we note that as our signal $I_1,\ldots,I_T$ are zero mean as defined in \eqref{eq:intef-It}. 
	\BE
	\Var{\left[\bnbC(\Dl)\right]}=\E{\left|\bnbC(\Dl)\right|^2}
	\EE
	Thus, we expand the second moment below.
	Using again the definition of the average correlation in \equref{eq:avg-auto-corr} of the main paper, we express:
%
	
	\BEA
	\!\!\!	\!\!\!&	\!\!\!\!\!\!&\E{| \bnbC(\Dl)|^2} = \frac{1}{T^2}\sum_{(t_1,t_2)} \E{\bC^{{I_{t_1}}}(\Dl) \cdot{\bC^{{I_{t_2}}}(\Dl)}^* }\\
			\!\!\!\!\!\!&	\!\!\!	\!\!\!&=\frac{1}{T^2}\sum_{t} \E{\bC^{{I_{t}}}(\Dl) \cdot{\bC^{{I_{t}}}(\Dl)}^* } \\
   	\!\!\!	\!\!\! &	\!\!\!	\!\!\!&+ \frac{1}{T^2}\sum_{(t_1\neq t_2)} \E{\bC^{{I_{t_1}}}(\Dl) }\cdot \E{{\bC^{{I_{t_2}}}(\Dl)} }^*\label{eq:indp-corr}\\
    	\!\!\!	\!\!\!&	\!\!\!	\!\!\!&=\frac{1}{T^2}\sum_{t} \E{|\bC^{{I_{t}}}(\Dl) |^2 }\label{eq:zero-corr} \\
    	\!\!\!	\!\!\!&	\!\!\!	\!\!\!&=\frac{1}{T} \E{|\bC^{{I_{1}}}(\Dl) |^2 }\label{eq:zero-corr-final}
	\EEA
	
	\noindent where \equref{eq:indp-corr} follows from the assumption that for $t_1\neq t_2$ the speckles $S^\bzero_{t_1},S^\bzero_{t_2}$ are uncorrelated with each other, and \equref{eq:zero-corr} from the fact that for displacements $\Dl\in \CListdl$ the correlation has zero expectation.
	
	From \equref{eq:zero-corr-final}, we conclude that as we increase the number of measurements the denominator scales as $1/T$. As a result the correlation contrast in \equref{eq:contrast-app} scales linearly with $T$.
\end{proof}

Below, we show that unlike the random modulation of \equref{eq:randS}, the \AlgName measurements of \equref{eq:ourS} do not reduce the correlation.

\begin{claim}
For displacements in the order of a few speckle grains, the correlation between \AlgName signals $S^{\oinp}_{t},S^{\tinp}_{t}$ produced by different illuminators $\oinp,\tinp$ is approximately the same as the correlation of the original speckle intensity images.	
\end{claim}
\begin{proof}
As we filter the intensity images to have zero mean, the expected correlation becomes the \emph{intensity covariance} $C_I$  defined in \cite{SeeThroughSubmission}:
\begin{align}
&C_I(\abs{u^{\oinp}(\snsp)}^2,\abs{u^{\tinp}(\snsp+\Dl)}^2) \\ &\equiv \E{|u^{\oinp}(\snsp)|^2|u^{\tinp}(\snsp+\Dl)|^2}-\E{|u^{\oinp}(\snsp)|^2}\E{|u^{\tinp}(\snsp+\Dl)|^2}
\end{align}
Where the expectation is taken other multiple tissue layers of the same type. Classical statistics results state that the covariance between intensities is the square of the covariance between the complex zero mean fields $u$. Thus, the above definition can be further simplified to:
\BE
C_I(\abs{u^{\oinp}(\snsp)}^2,\abs{u^{\tinp}(\snsp+\Dl)}^2) = \left|\E{u^{\oinp}(\snsp){u^{\tinp}(\snsp+\Dl)}^*}\right|^2
\EE
On the other hand, we consider the correlation between \AlgName signals:
\begin{align}
&C_I(S^{\oinp}_t(\snsp),S^{\tinp}_t(\snsp+\Dl)) \\
&=\E{S^{\oinp}_t(\snsp)S^{\tinp}_t(\snsp+\Dl)}-\E{S^{\oinp}_t(\snsp)}\E{S^{\tinp}_t(\snsp+\Dl)}
\end{align}
As before, the second term vanishes as in  \AlgName, $S^{\inp^n}_t(\snsp)$ is a zero mean signal. Also as we assume that different pixels of a speckle pattern are independent, we can express:
\begin{align}
&\E{S^{\oinp}_t(\snsp)S^{\tinp}_t(\snsp+\Dl)} \\
&=\E{ u^{\oinp}(\snsp)u^{\oinp}(\snsp+ \pt)^* {u^{\tinp}(\snsp+\Dl)}^*u^{\tinp}(\snsp+\Dl+ \pt)}\\\nonumber
&=\E{ u^{\oinp}(\snsp) {u^{\tinp}(\snsp+\Dl)}^*} \E{ u^{\oinp}(\snsp+ \pt) u^{\tinp}(\snsp+\Dl+ \pt)^*}^*.
\end{align}
Now we assume the displacement $\pt$ is modest enough so that the correlation at pixel $\snsp$ and the correlation at a small displacement, at pixel $\snsp+\pt$ is similar, so that
\BE\label{eq:cor-local-smooth}
\E{ u^{\oinp}(\snsp) {u^{\tinp}(\snsp+\Dl)}^*}\approx
\E{ u^{\oinp}(\snsp+ \pt) u^{\tinp}(\snsp+\Dl+ \pt)^*}.
\EE
Put all together, we can derive
\begin{align}
C_I(S^{\oinp}_t(\snsp),S^{\tinp}_t(\snsp+\Dl)) & \approx  \left|\E{u^{\oinp}(\snsp){u^{\tinp}(\snsp+\Dl)}^*}\right|^2 \\
& = C_I(\abs{u^{\oinp}(\snsp)}^2,\abs{u^{\tinp}(\snsp+\Dl)}^2)
\end{align}
\end{proof}

We now show that despite the fact that we capture multiple images through the same tissue layer and our measurements are not independent, they are still uncorrelated which is enough to reduce the noise of the speckle auto-correlation we evaluate.
\begin{claim}
For displacements ${\opt},{\tpt} $ whose distance $\|{\opt}-{\tpt}\|$ is larger than the speckle grain, the signals $S^{\inp^n}_{t_1},S^{\inp^n}_{t_2}$ are uncorrelated, so that
	\BE\label{eq:uncor-disp-spk}
	\E{S^{\inp^n}_{t_1}\cdot {S^{\inp^n}_{t_2}}^*}-\E{S^{\inp^n}_{t_1}}\cdot\E{ S^{\inp^n}_{t_2}}^*=0
	\EE
\end{claim}
\begin{proof}
	Our derivation is based on the assumption that the speckle fields  $u^{\inp^n}$ have zero means, and the speckle values in different sensor positions $\snsp$ are independent random variables.
	For a non zero displacement $\pt$ we get
	\BEA
\E{S^{\inp^n}_{t}} &=& \E{u^{\inp^n}(\snsp)u^{\inp^n}(\snsp+\pt)}  \nonumber \\
&=& \E{u^{\inp^n}(\snsp)}\E{u^{\inp^n}(\snsp+\pt)} = 0. \label{eq:-indp-mean}
\EEA
In  a similar way
\begin{align}
\E{S^{\inp^n}_{t_1}{S^{\inp^n}_{t_2}}^*}\! &= \E{u^{\inp^n}(\snsp){u^{\inp^n}(\snsp+\opt)}^*{u^{\inp^n}(\snsp)}^*u^{\inp^n}(\snsp+\tpt)}  \\
&=\E{|u^{\inp^n}(\snsp)|^2 u^{\inp^n}(\snsp+\opt)^*u^{\inp^n}(\snsp+\tpt)}  \\
&=\E{|u^{\inp^n}(\snsp)|^2}\E{u^{\inp^n}(\snsp+\opt)}^*\!\E{u^{\inp^n}(\snsp+\tpt)}\label{eq:-indp-3}\\
&= 0 \label{eq:indp-3-0} .
\end{align}

Where \equref{eq:-indp-3}  follows again from the assumption that speckle at different pixel positions are independent.
\equpref{eq:-indp-mean}{eq:indp-3-0} prove the desired \equref{eq:uncor-disp-spk}.
\end{proof}

We now move to study the relationship between the correlation we can measure from \AlgName to the actual correlation of the latent pattern.

\begin{claim}\label{clm:ramp}
	For speckle fields $u^{\oinp}(\snsp),u^{\tinp}(\snsp)$ satisfying the tilt shift relationship of \equref{eq:me-corr-tilt-shift}, the \AlgName  measurements $	S^{\oinp}_t(\snsp)$ and 	$S^{\tinp}_t(\snsp)$ defined in \eqref{eq:ourS}  are shifted versions of each other, times a {\em globally} constant phasor, which is independent of pixel position $\snsp$. This is expressed by the relationship:
	\BE\label{eq:phasor-for-disp}
	S^{\oinp}_t(\snsp)\approx	S^{\tinp}_t(\snsp+\Dl)e^{-jk\alpha<\Dl,\pt>},
	\EE with $\Dl=\tinp-\oinp$.	
\end{claim}

\begin{proof}
From \equref{eq:me-corr-tilt-shift}, we have $u^{\oinp}(\snsp)\approx e^{jk \alpha <\Dl,\btau>} u^{\tinp}(\snsp+\Dl)$, with $\btau=\snsp-\oinp$. Substituting into \eqref{eq:ourS} we have,

\BEA
S^{\oinp}_t(\snsp) &=& u^{\oinp}(\snsp)u^{\oinp}(\snsp+ \pt)^* \nonumber  \\
&\approx& e^{jk \alpha <\Dl,\btau>} u^{\tinp}(\snsp+\Dl) e^{-jk \alpha <\Dl,\btau+\pt>} u^{\tinp}(\snsp+\Dl+\pt)^* \nonumber \\
&=& e^{-jk \alpha <\Dl,\pt>} u^{\tinp}(\snsp+\Dl) u^{\tinp}(\snsp+\Dl+\pt)^* \nonumber \\
&=& e^{-jk \alpha <\Dl,\pt>} S^{\tinp}_t(\snsp+\Dl) \nonumber
\EEA
\end{proof}

With all the above claims we are now ready to prove our main result and show that the global phasor of the previous claim translates into a phase ramp in the auto-correlation.

\begin{claim}\label{claim:trans-interf-auto-corr-app}
	Using the \AlgName measurements of \equref{eq:ourS}, the speckle auto-correlation  $\bC^{I_t}=I_t\star I_t$ is equivalent to the auto-correlation of the latent pattern $\bC^O=O\star O$,  times a phase ramp correction
	\BE\label{eq:ac-I-tilt-ac-O}
	\bC^{I_t}(\Dl)\approx \bC^{O}(\Dl)e^{-jk\alpha<\Dl,\pt>}.
	\EE
\end{claim}
\newcommand{\addFigurea}[1]{
\begin{subfigure}{0.59\linewidth}
    \centering
    \smallskip
    {\includegraphics[width=1\textwidth,height=1\textwidth]{#1}}
 \end{subfigure}}
\newcommand{\addFigureSmalla}[1]{
\begin{subfigure}{0.3\linewidth}
    \centering
    \smallskip
    {\includegraphics[width=1\textwidth]{#1}}
 \end{subfigure}
}
\newcommand{\littleSpacea}{\hspace{-0.05\linewidth}}
\newcommand{\vlittleSpacea}{\vspace{-0.05\linewidth}}

\begin{figure}[t]
        \centering
        \begin{tabular}{cc}
        &\littleSpacea Full Frame Correlation               
        \\
        \littleSpacea\begin{tabular}{cc}
        \rotatebox[origin=c]{90}{Ground Truth}&
                \addFigureSmalla{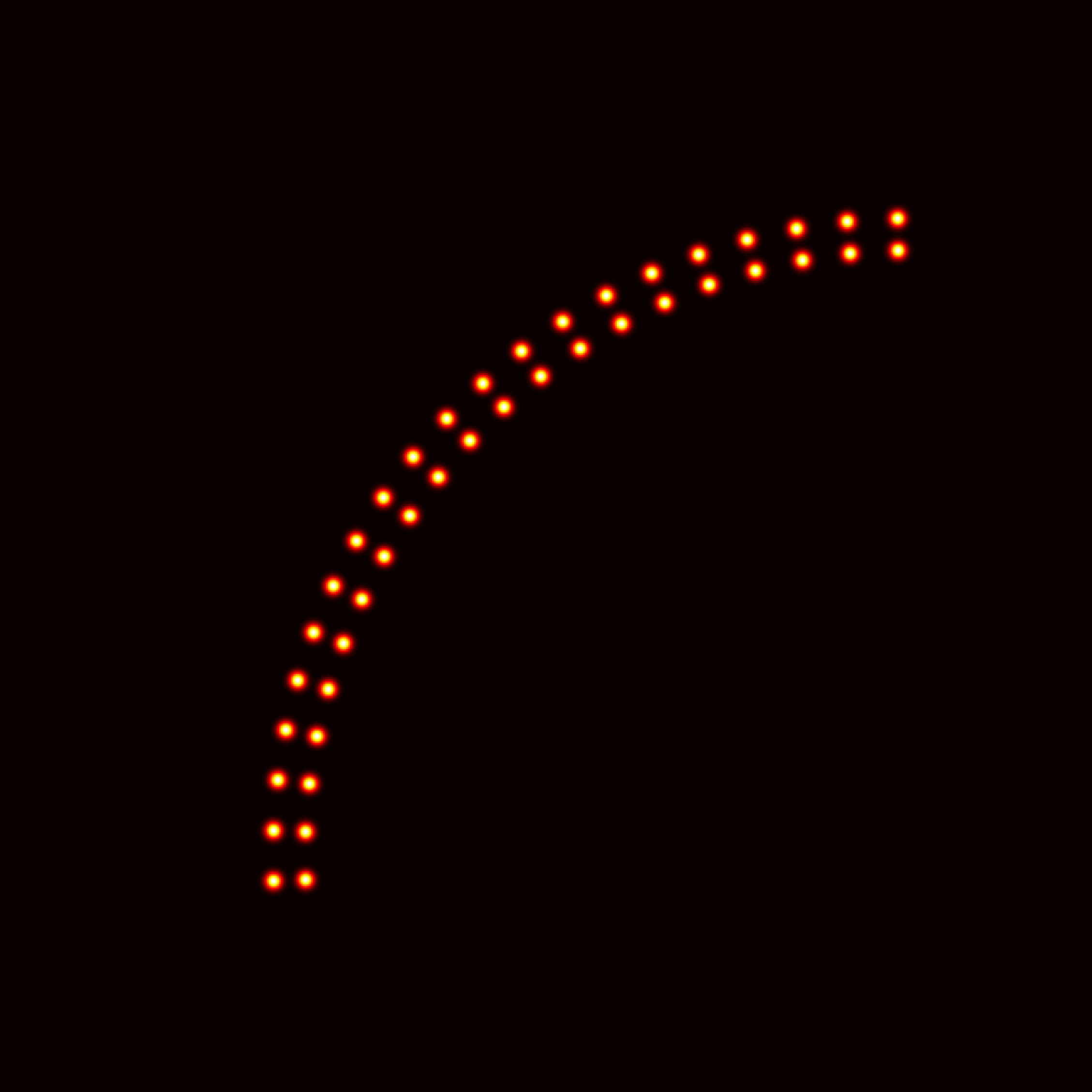}\\
                \rotatebox[origin=c]{90}{Input}&
                \addFigureSmalla{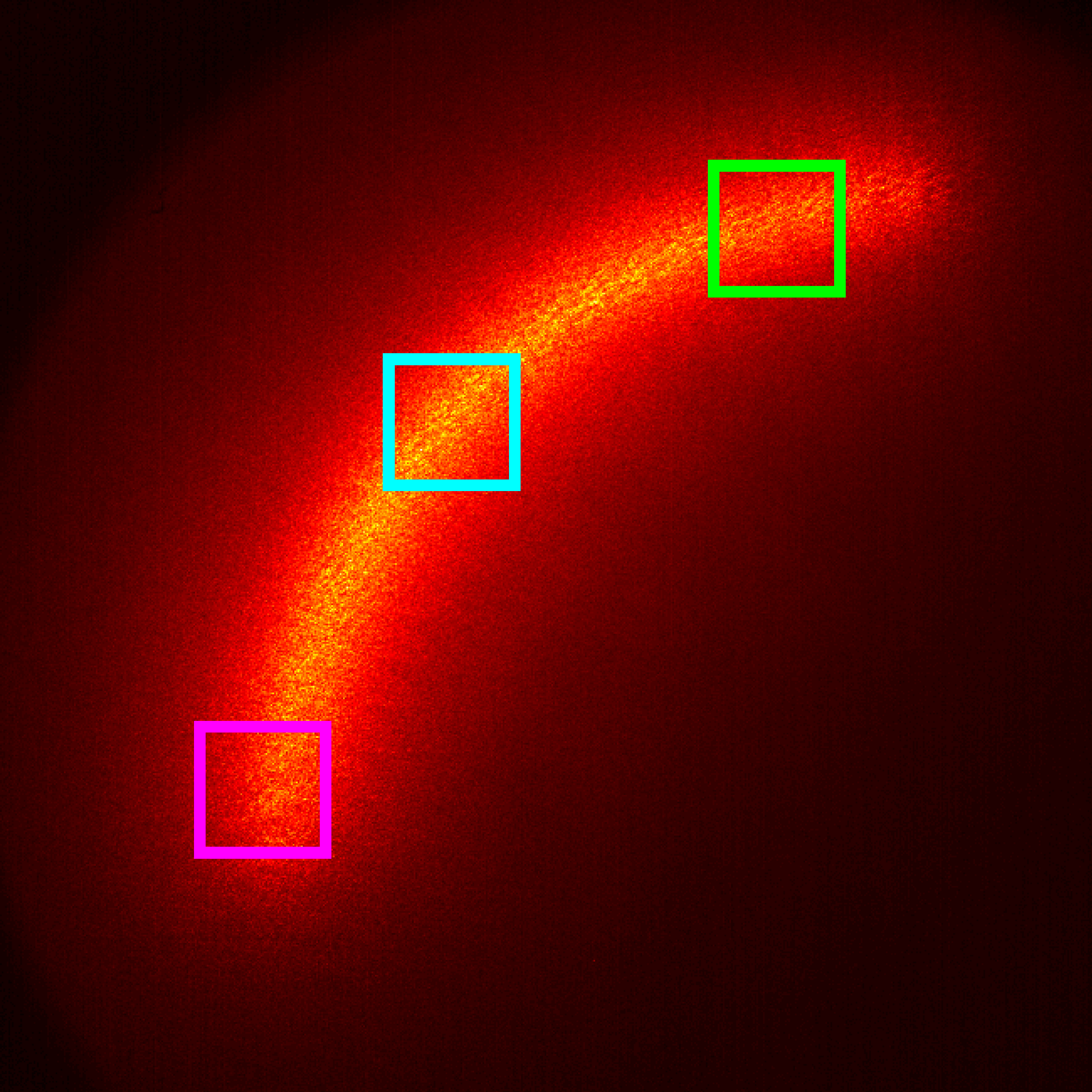}
             \end{tabular}&
             \hspace{-0.02\linewidth}\littleSpacea\addFigurea{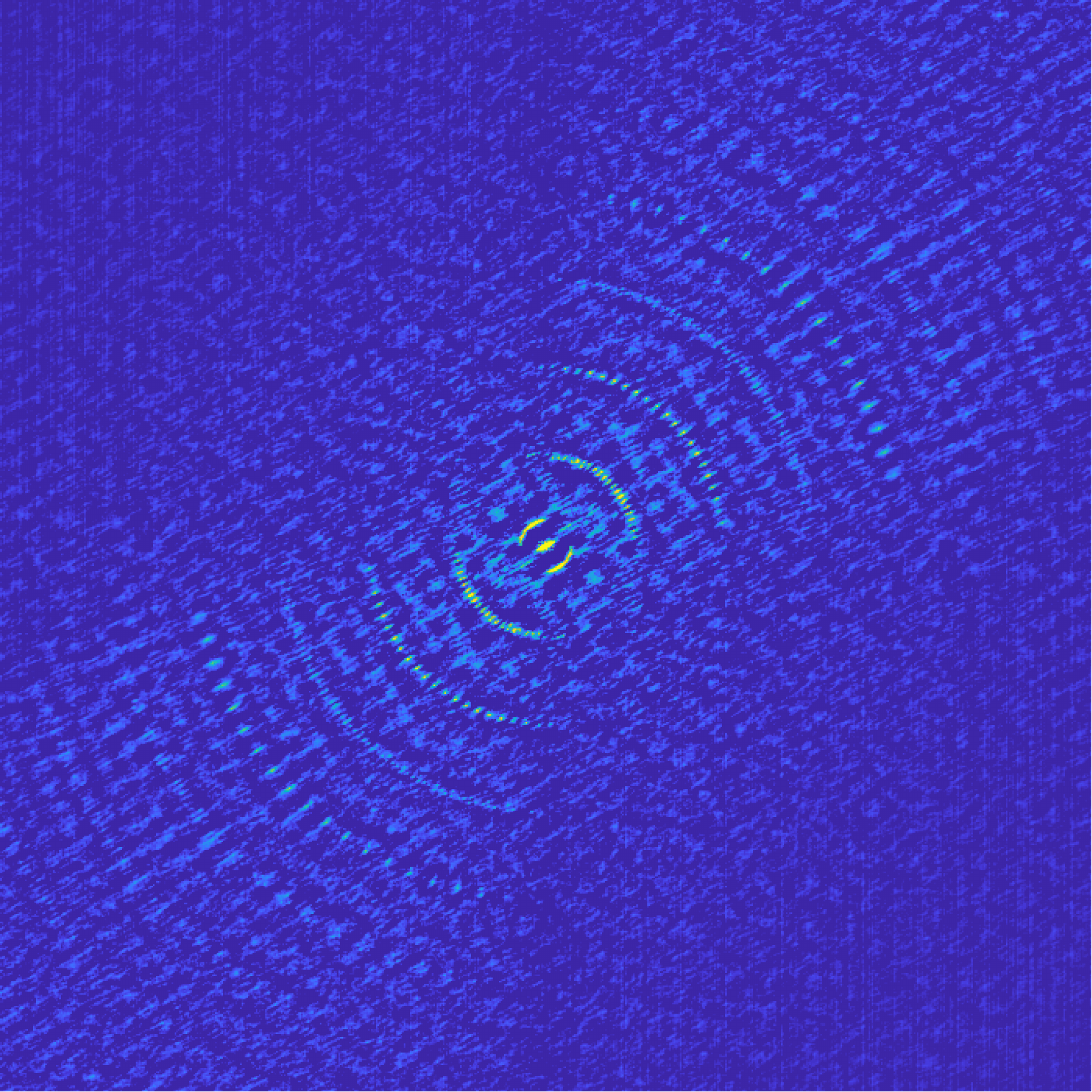}
        \end{tabular}\\
        \vspace{0.02\linewidth}
        \begin{tabular}{cccc}
        \hspace{-0.02\linewidth}\rotatebox[origin=c]{90}{Local Correlation}&
           \addFigureSmalla{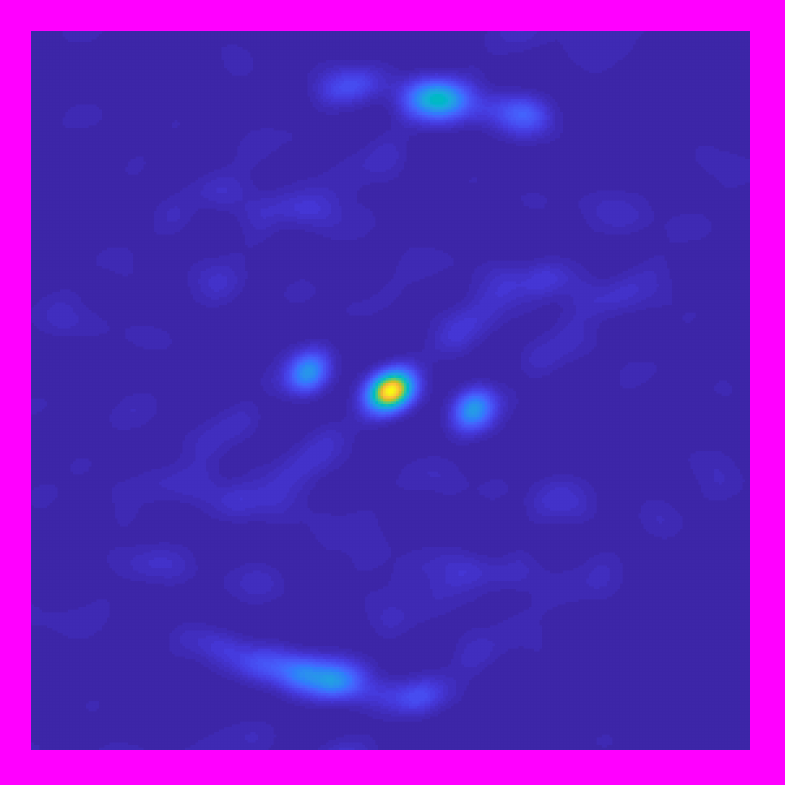}&
         \littleSpacea\addFigureSmalla{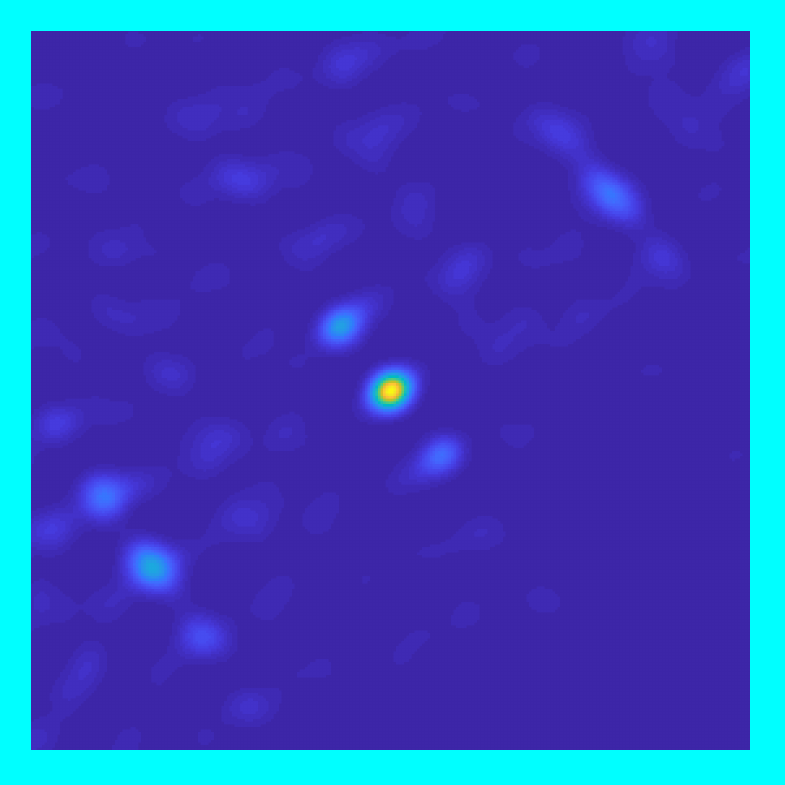}&
         \littleSpacea\addFigureSmalla{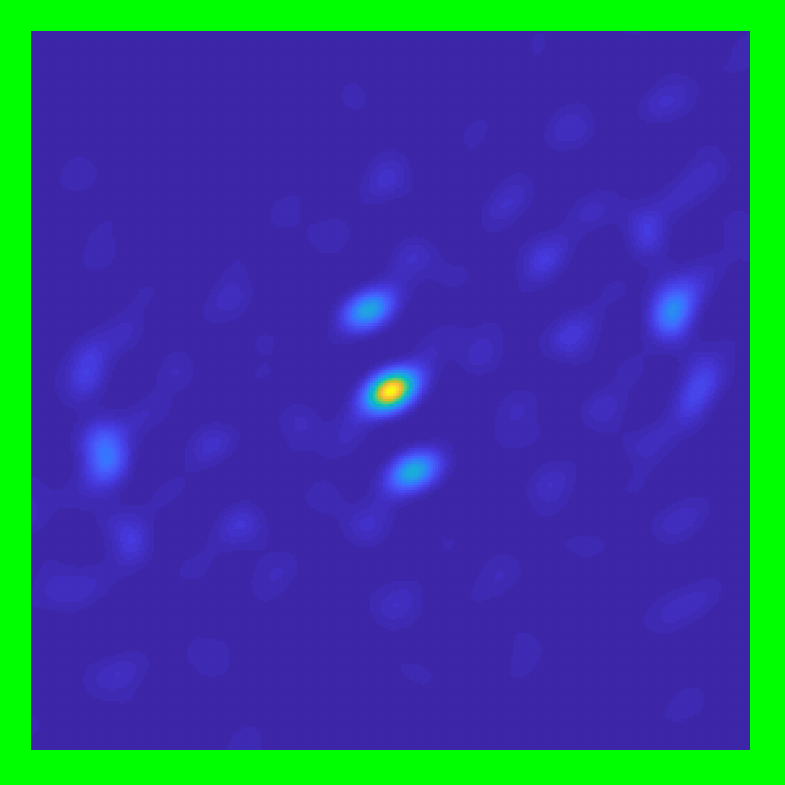} 
        \end{tabular}
     \caption{{\bf Local versus global auto-correlation.} The orientation of the auto-correlation evaluated in three different local windows of the image matches the orientation of the arc in the corresponding region of the latent image. By contrast, the auto-correlation of the full frame is much nosier, and  decays for large displacements due to limited ME.}
        \label{fig:arc_win}
\end{figure}


\begin{proof}
Using \clmref{clm:ramp} for the specific case $\oinp = \inp^n$ and $\tinp = \bzero$ we get $S^{\inp^n}_t(\snsp)\approx \So_t(\snsp+\inp^n)e^{-jk\alpha<\inp^n,\pt>}$. Summing over all sources $S^{\inp^n}_t(\snsp)$ we have
\BE
I_t(\snsp) \approx  \So_t \ast \tilde{O}. \nonumber
\EE
with $\tilde{O}$
\BE\label{eq:tild-O-pr}
\tilde{O}(\snsp)=O(\snsp) e^{-jk\alpha<\snsp,\pt>}.
\EE
This is due to the fact that $O$ is non zero only at positions $\snsp=\inp^n$ for one of the sources, so effectively $\tilde{O}$ has for each sensor position  $\inp^n$ the global phasor of \equref{eq:phasor-for-disp}.
As in the standard derivation of the speckle auto-correlation we assume
$\So_t\star \So_t=\delta$, and hence
\BE
I_t \star I_t \approx \tilde{O} \star \tilde{O}.
\EE
or equivalently
\BE\label{eq:ca-I-to-ac-O}
	\bC^{I_t}(\Dl)\approx \bC^{\tilde{O}}(\Dl).
\EE
Hence we are left with the need of computing $\bC^{\tilde{O}}(\Dl)$.
We note that by the  Wiener-Khinchin theorem, $\bC^{\tilde{O}}(\Dl)$ is the inverse Fourier transform of $\abs{\fourier(\tilde{O})}^2$.
However as $\tilde{O}$ is obtained by multiplying $O$ with a phase ramp (see \equref{eq:tild-O-pr}), their Fourier transforms are related via a shift:
\BE
\fourier(\tilde{O})(\bomg)=\fourier({O})(\bomg+\alpha\pt).
\EE
The shift relation holds also for their absolute values
\BE
\abs{\fourier(\tilde{O})(\bomg)}^2=\abs{\fourier({O})(\bomg+\alpha\pt)}^2
\EE
Hence $\bC^{\tilde{O}}(\Dl)$ and $\bC^{{O}}(\Dl)$ are related via a tilt:
\BE\label{eq-tild-O-ac}
\bC^{\tilde{O}}(\Dl)=\bC^{{O}}(\Dl)e^{-jk\alpha<\Dl,\pt>}.
\EE
Substituting \equref{eq-tild-O-ac} in \equref{eq:ca-I-to-ac-O} proves the desired \equref{eq:ac-I-tilt-ac-O}.

\end{proof}

\vspace{0.05in}

\subsection{Optimizing using local support}\label{sec:optimization}


For sources located inside the scattering medium, speckle patterns emerging from a single source have local support and do not spread over the entire sensor. To take advantage of this property,~\cite{SeeThroughSubmission} suggest matching the local speckle correlations in the image, rather than the full-frame auto-correlation. We review this algorithm below.

\newcommand{\addFigureo}[1]{
\begin{subfigure}{0.36\linewidth}
    \centering
    \smallskip
    {\includegraphics[width=1\textwidth,height=0.89\textwidth]{#1}}
 \end{subfigure}}
\newcommand{\addFigurec}[1]{
	\begin{subfigure}{0.23\linewidth}
		\centering
		\smallskip
		{\includegraphics[width=1\textwidth,height=1\textwidth]{#1}}
\end{subfigure}}

\newcommand{\littleSpacev}{\vspace{-0.18cm}}
\newcommand{\littleSpaceo}{\hspace{-0.02\linewidth}}
\newcommand{\bigSpaceo}{\hspace{0.02\linewidth}}
\newcommand{\bigSpacev}{\vspace{0.05cm}}

\begin{figure}[!t]
        \centering
        \begin{tabular}{c}
        \begin{tabular}{cc}
        $I$&
        $O$
        \\
        \addFigureo{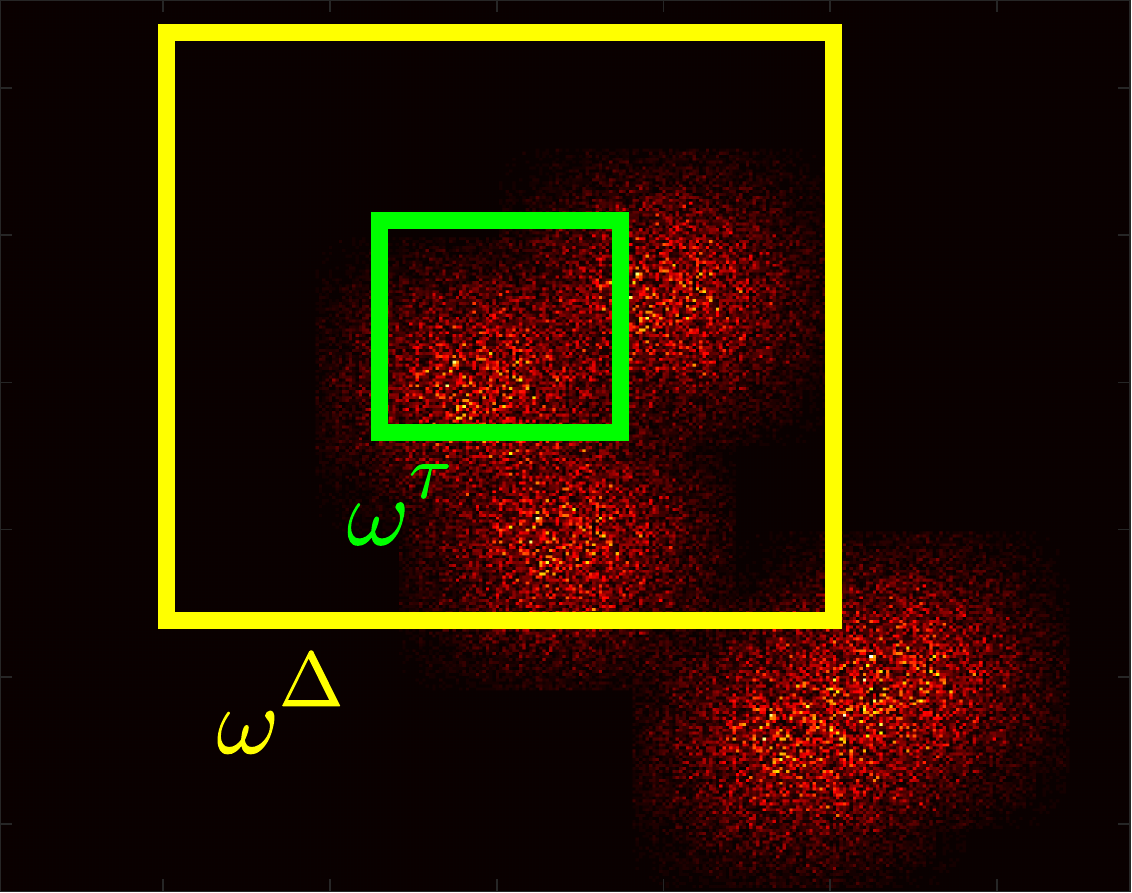}&\littleSpaceo
        \addFigureo{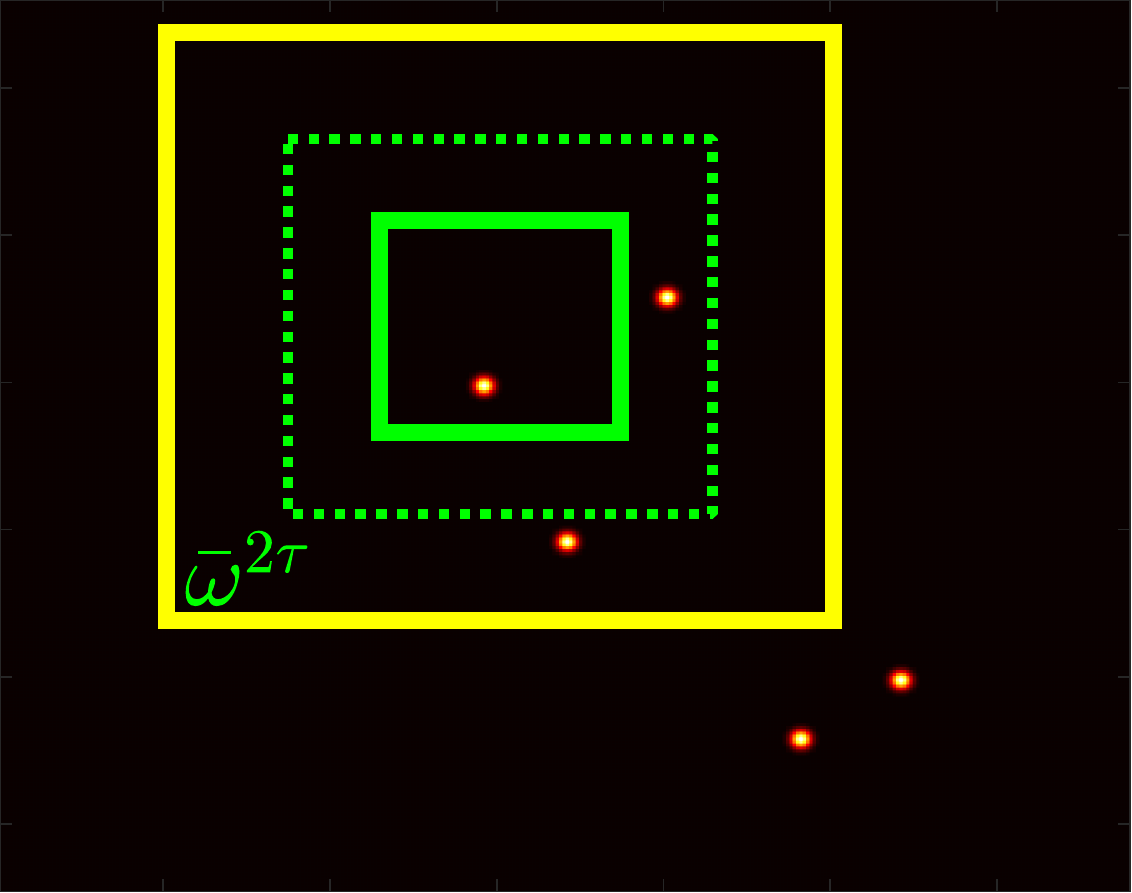}\\
    \end{tabular}
\\[1.7 cm]
 \begin{tabular}{ccc}
 	 $I_{w^\tau}\star I_{w^\Delta}$&
 	$O_{w^\tau}\star O_{w^\Delta}$&
 	$O_{\brw^{2\tau}}\star O_{w^\Delta}$
 	\\
        \addFigurec{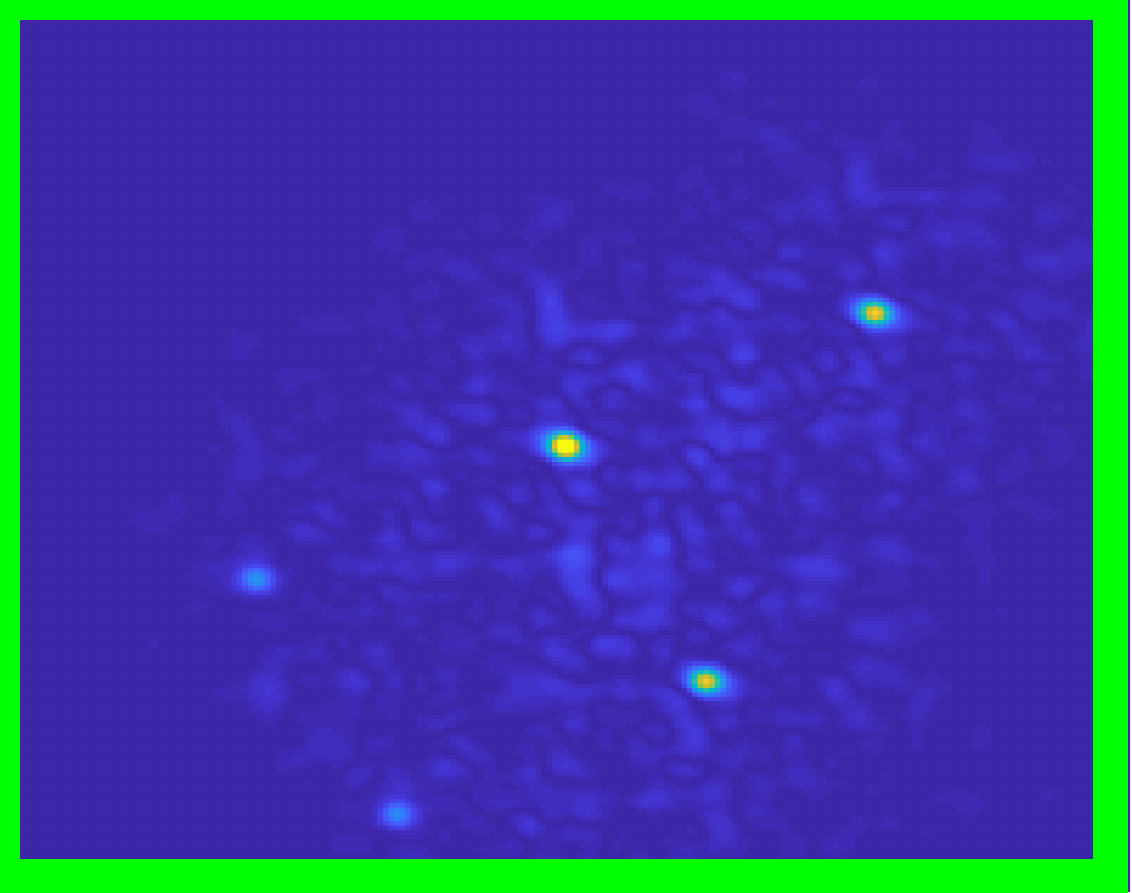}&\littleSpaceo
        \addFigurec{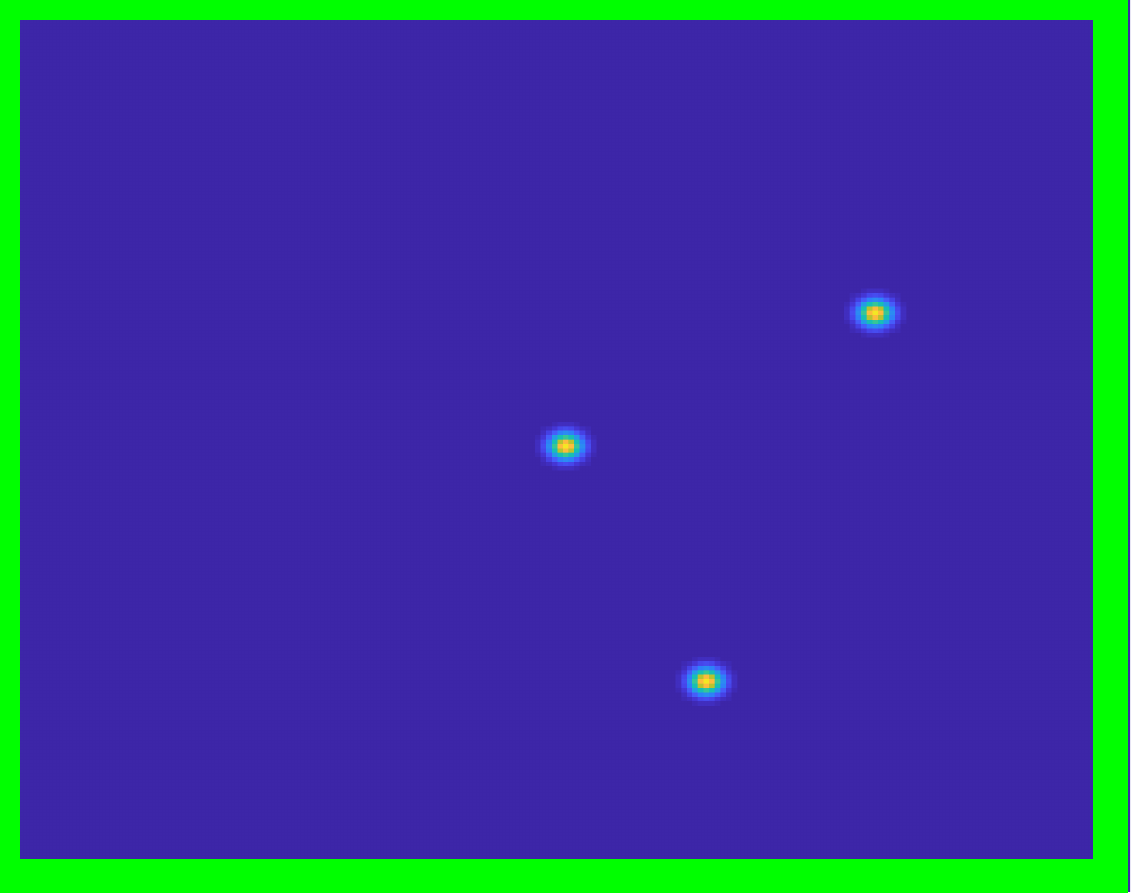}&\littleSpaceo
        \addFigurec{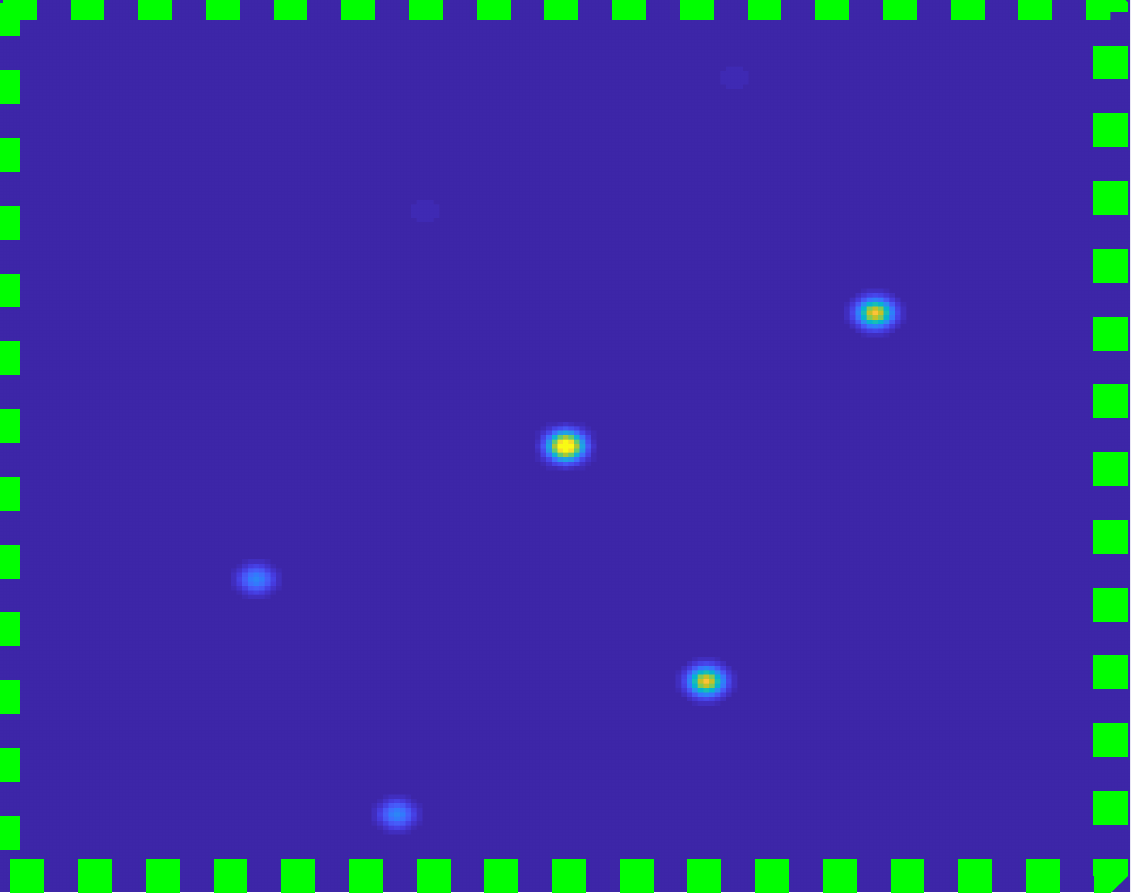}\littleSpaceo
         \end{tabular}
     \\[1.2 cm]
         \begin{tabular}{cc}
         	 $I$&$O$
         	 \\
        \addFigureo{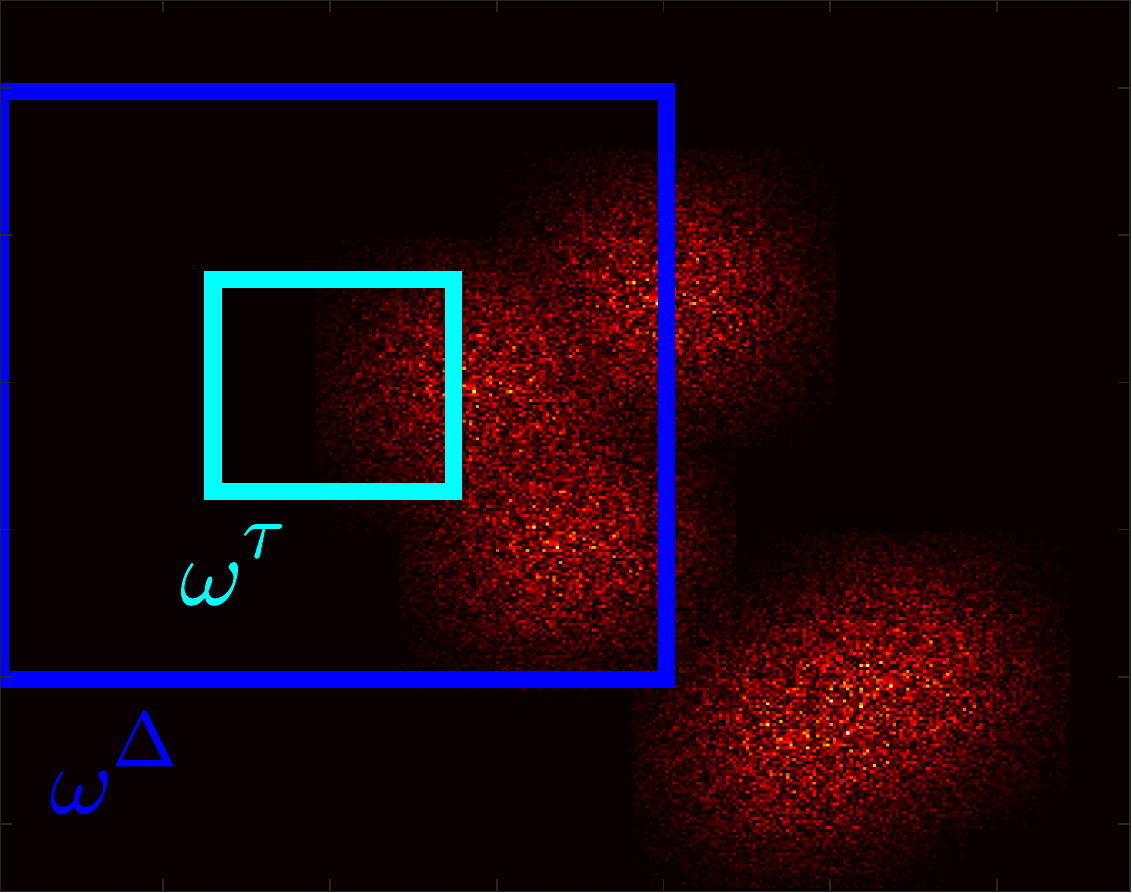}&\littleSpaceo
		\addFigureo{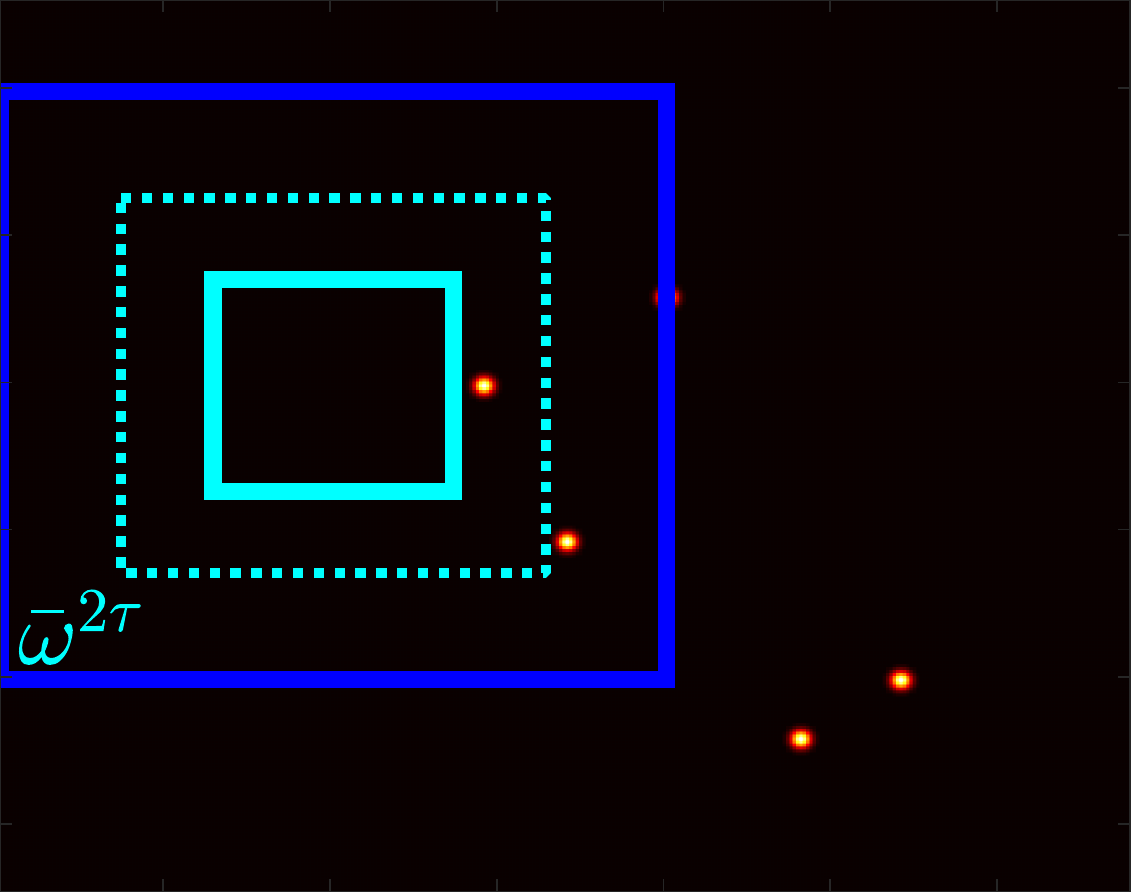}
	\end{tabular}
\\[1.7 cm]
\begin{tabular}{ccc}
	 $I_{w^\tau}\star I_{w^\Delta}$&
	$O_{w^\tau}\star O_{w^\Delta}$&
	$O_{\brw^{2\tau}}\star O_{w^\Delta}$ 
	\\
		\addFigurec{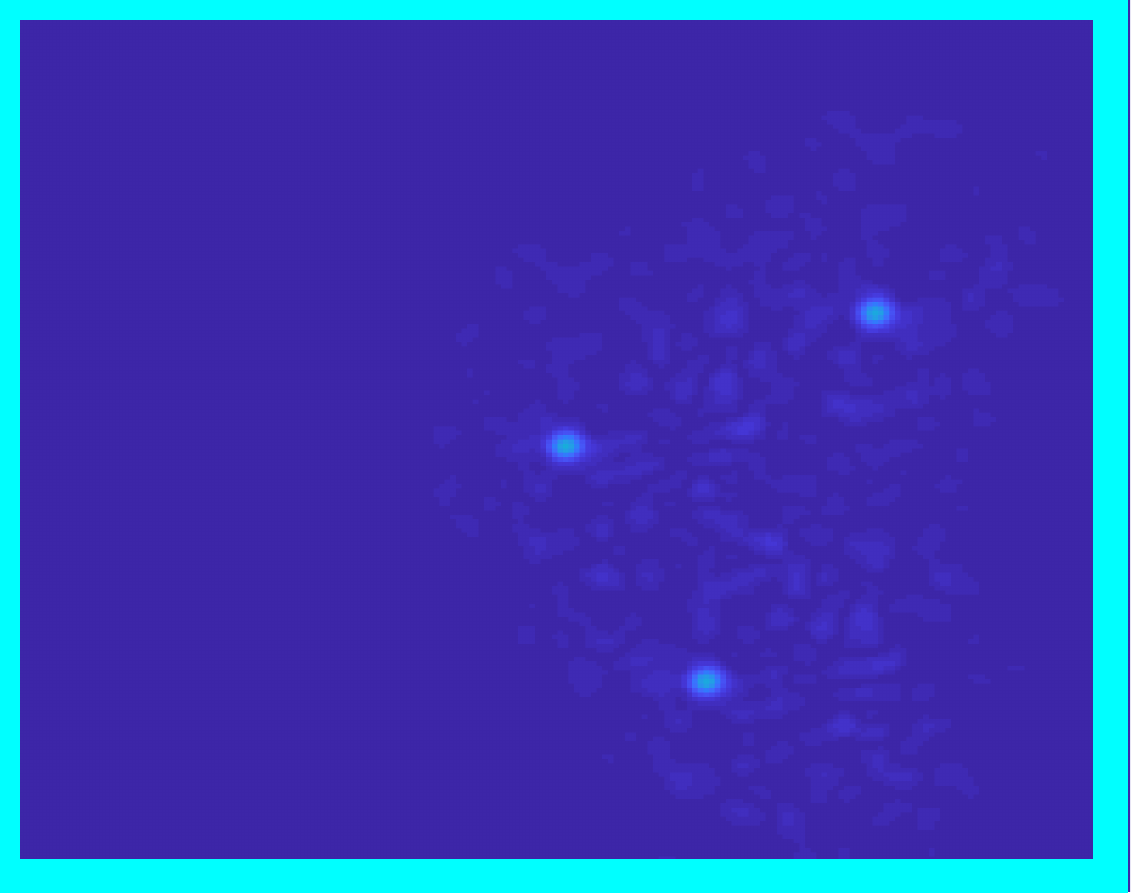}&\littleSpaceo
		\addFigurec{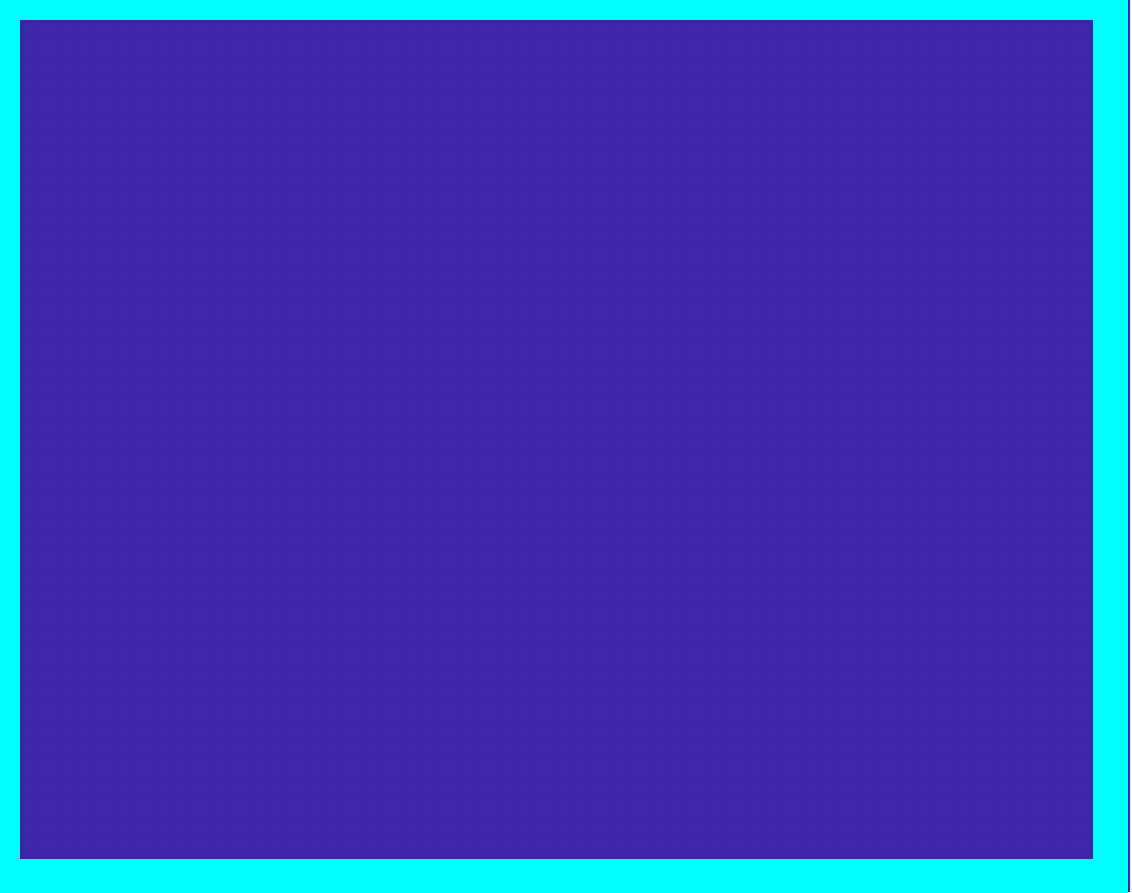}&\littleSpaceo	
		\addFigurec{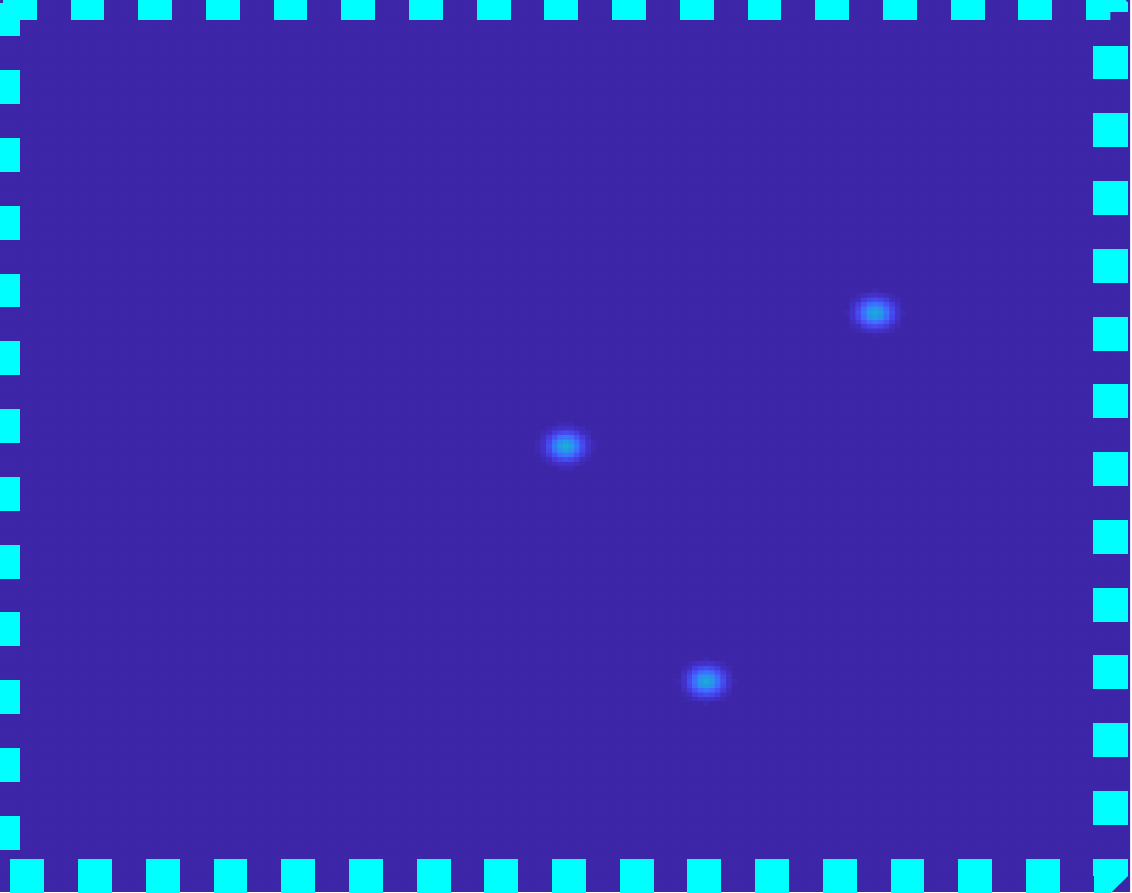}\littleSpaceo
      \end{tabular}
   \end{tabular}	\vspace{-0.0cm}
     \caption{{\bf Local window selection for optimization.} We consider local subwindows $w^\tau$ (light green and cyan frames) whose support is equivalent to the speckle support size. Each such window is correlated with a wider window $w^\Delta$ (yellow and blue frames) around it, whose support is equivalent to the ME range. As speckle inside window $w^\tau$ can arise from a source outside $w^\tau$,  $O_{w^\tau}\star O_{w^\Delta}$ may not match  $I_{w^\tau}\star I_{w^\Delta}$. To overcome this, we use an extended non-binary sub-window $\brw^{2\tau}=w^\tau\star w^\tau$ for $O$, whose support is indicated by dashed lines.}
        \label{fig:optim_win}
\end{figure}

\begin{figure}[t]
\centering
\includegraphics[width=\linewidth]{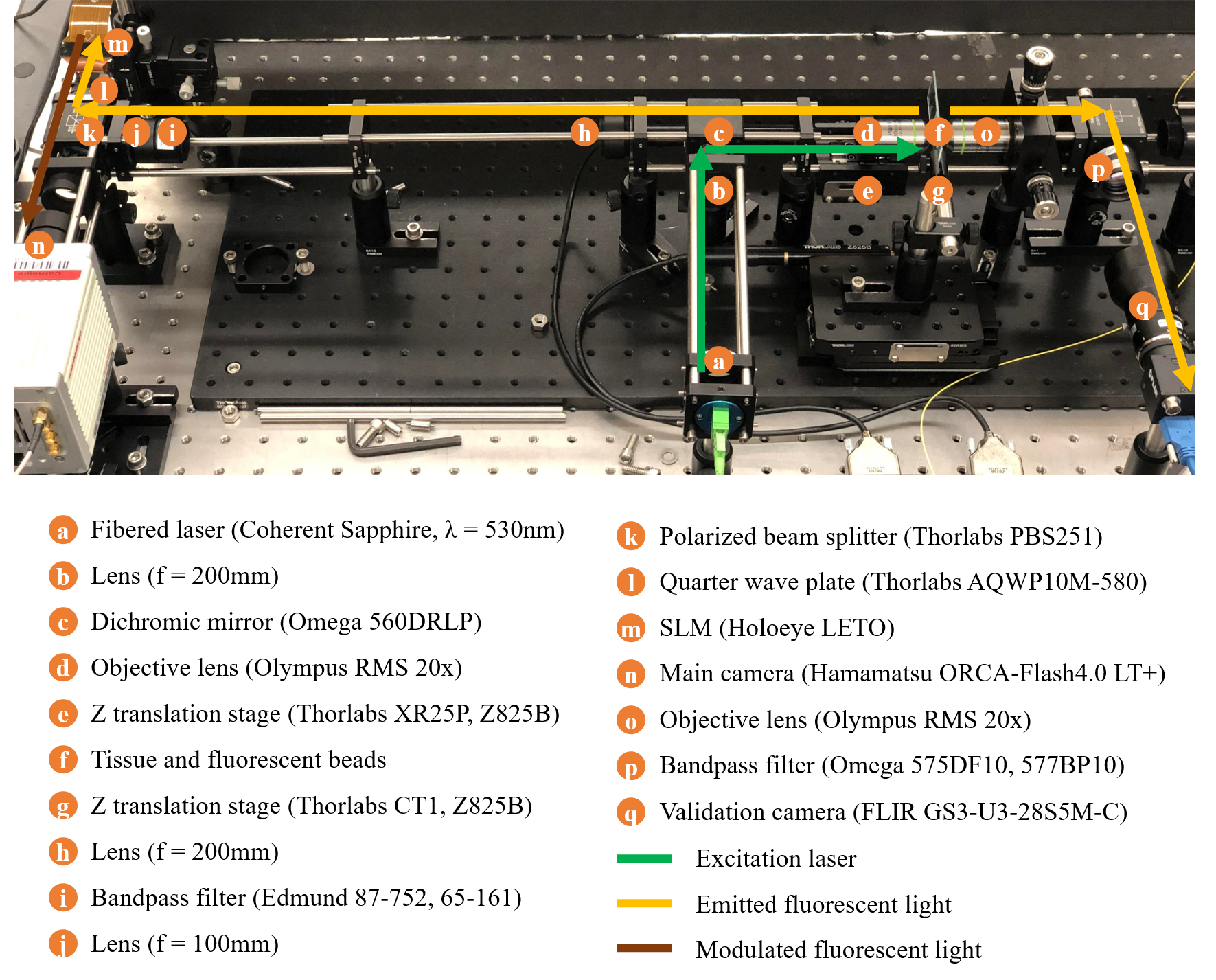}
\caption{\label{fig:prototype}\textbf{Hardware prototype.} Image of our prototype, anal-
ogous to the schematic in \figref{fig:setup}. }
\end{figure} 

\begin{figure*}[t]
\centering
\includegraphics[width=55em]{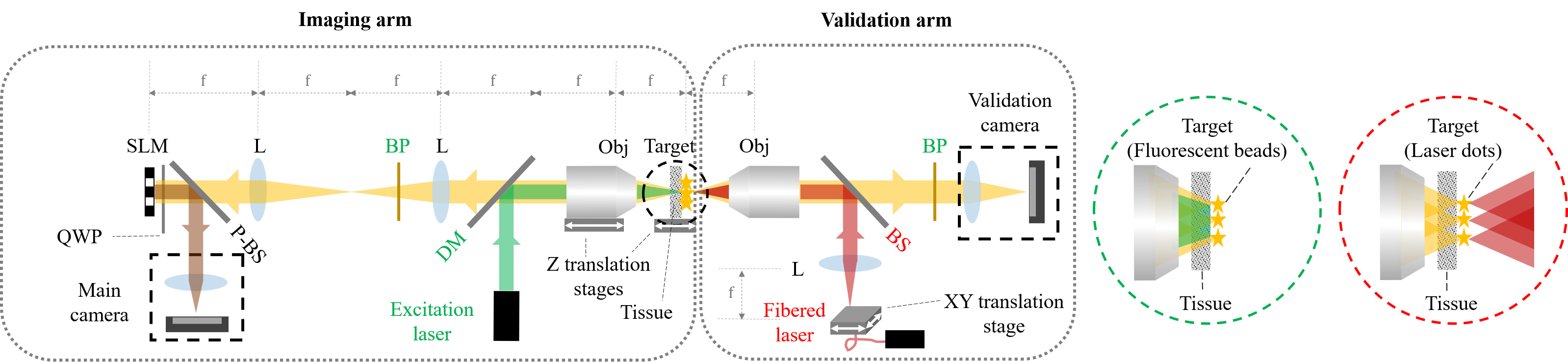}
\caption{\label{fig:setupfull}\textbf{Full experiment setup combining fluorescent beads with a shifting laser.} We add to the setup in \figref{fig:setup} a  fibered laser source. The diffused fibered output is imaged using a tube lens and objective to generate a source exactly at the back plane of a tissue layer. The laser light is scattered through the tissue and the speckle pattern it generates is imaged by a the main camera from the other side of the tissue.   We mount a fibered laser on a 2D translation stage so we can create programmable patterns.}
\end{figure*} 

For motivation, consider \figref{fig:arc_win} that we re-plot from~\cite{SeeThroughSubmission}.
It visualizes speckles produced by latent incoherent illuminators in a double arc layout.
 Computing auto-correlation at small subwindows of the speckle image reveals the local orientation of the arc in the latent image. By contrast, when computing the auto-correlation of the full frame, the correlation is considerably noisier even for small displacements. Correlations between far illuminators are even harder to detect due to the fact that the ME range is limited and for large displacements $\Dl$ the desired correlation (see \equref{eq:me-corr-tilt-shift}) is very weak.

The optimization algorithm takes as input two threshold parameters  $\maxt,\maxd$. It
assumes that speckles from one illuminator are spread over pixels in a window of size $\maxt$ around it, and that ME correlation holds for displacements  $|\Dl|<\maxd$.
The thresholds $\maxt,\maxd$ are free parameters that can be fine-tuned to improve reconstruction quality, and~\cite{SeeThroughSubmission} show that  performance  are not too sensitive  to their exact values.
The algorithm offers improved performance compared to the baseline full-frame auto-correlation algorithm in situations where $\maxt<\maxd$, namely when the support from one illuminator is lower than the ME range. For thick scattering slices, where high-order scattering is dominant, this relationship does not hold and the approach reduces to the baseline full-frame auto-correlation algorithm.

The algorithm searches for a latent image $O$ such that the auto-correlation in its local windows will match the auto-correlation in the local windows of the input image $I$. We define $w^\Dl$ and $w^\btau$ to be binary windows with support $\maxd,\maxt$, respectively, and $\brw^{2\btau}=w^\btau\star w^\btau$---note that, from its definition, $\brw^{2\btau}$ is non-binary. Then, we recover $O$ by solving the optimization problem:
\BE\label{eq:minO}
\min_O \textstyle\sum_p \|\avgt e^{jk\alpha<\Dl,\pt>}\cdot I_{t,w^\btau_p}\star I_{t,w^\Dl_p}- O_{\brw^{2\btau}_p}\star O_{w^\Dl_p}\|^2,
\EE
where $I_{t,w^\btau_p},I_{t,w^\Dl_p},O_{\brw^{2\btau}_p},O_{w^\Dl_p}$ denote windows of a given size cropped from the input and latent images, centered around the $p$-th pixel.

\equref{eq:minO} uses windows of three different sizes, and we use \figref{fig:optim_win} to visualize their different roles: Each $w^\btau_p$ is a small window around pixel $p$ whose support is equivalent to the expected support size of the speckle pattern due to a single illuminator.  $w^\Dl_p$ is a larger window around the same pixel, corresponding to the maximal displacement $\maxd$ for which we expect to find correlation, as dictated by the ME range.

We note, additionally, that the window cropped from $O$ should be wider than that from $I$. This is because speckle at a certain pixel can arise from an illuminator within a window around it. For example, in \figref{fig:optim_win}, no illuminator is located inside  the cyan subwindow of $O$, but part of the speckle pattern of a neighboring source is leaking into the corresponding cyan subwindow of $I$. As a result $O_{w^\btau_p}\star O_{w^\Dl_p}$ is a zero image, even though $I_{w^\btau_p}\star I_{w^\Dl_p}$ detects three impulses. It is easy to prove that this can be addressed using the larger, non-binary window $\brw^{2\btau}$ in the latent image, indicated in \figref{fig:optim_win} using dashed lines: in this case, $O_{\brw^{2\btau}_p}\star O_{w^\Dl_p}$ correctly detects the same three impulses as $I_{w^\btau_p}\star I_{w^\Dl_p}$.

The motivation for the cost of \equref{eq:minO} is that, even if two illuminators in the latent pattern $O$ are at a distance larger than the ME range $\maxd$, they can be recovered if there exists a sequence of illuminators between them, where each two consecutive illuminators in the sequence are separated by a distance smaller than $\maxd$. For example, in \figref{fig:optim_win}, the illuminators outside the yellow and cyan $w^\Dl$ windows are recovered thanks to the intermediate illuminators.

The optimization problem in \equref{eq:minO} is no longer a phase retrieval problem as in standard full-frame auto-correlation algorithms. We minimize it  using the ADAM gradient-based optimizer~\cite{Kingma2014AdamAM}.
Gradient evaluation is described in \cite{SeeThroughSubmission}, and reduces to a sequence of convolution operations that can be performed efficiently, e.g., using a GPU based fast Fourier transform. For initialization, we set the latent image to random noise; we have observed empirically that the optimization is fairly insensitive to initialization. Finally, we note that even though we could place a window $w_p$ around every pixel of $I$, the empirical correlation  is insensitive to small displacements of the central pixel $p$. Therefore, in practice, we consider windows only at strides $\maxt/2$, which helps reduce  computational complexity.

We note that the optimization problem of \equref{eq:minO} is similar to ptychography algorithms~\cite{PhysRevLett.98.034801}. However, we emphasize that previous ptychographic approaches for extending the ME range recover the latent illuminators from \emph{multiple} image measurements, captured by sequentially exciting different areas on the scattering sample~\cite{zhou2019retrieval,Gardner:19,Li:2019:Ptycho,Li:B:2019:Ptycho,Shekel_2020}. By contrast, our algorithm recovers the latent illuminators from a fixed number of full-frame shots.




\subsection{Interferometric measurements}\label{sec:interf}
While most of our setup is similar to the one used by \cite{SeeThroughSubmission},
we use an SLM (Holoeye LETO) in the Fourier plane of the imaging arm, which we use to modulate the scattered light. We visualize the hardware schematic in \figref{fig:setup} and its image  in  \figref{fig:prototype}.
To capture interferometric  measurements, we use a polarizing beamsplitter to horizontally polarize the wave, followed by a quarter waveplate at an angle of $45^\circ$ to induce a $\pi/2$ phase delay along one of the axis to produce a circularly polarized light. The SLM only modulates the polarization state along its fast axis, which is horizontal, and its slow axis is reflected without any modulation. The light reflected off the SLM is sent again through quarter waveplate, which adds another $\pi/2$ phase shift that makes the light linearly polarized again, but in a vertical direction; finally, the polarizing beamsplitter interferes the modulated and unmodulated  waves.

To capture the \AlgName measurements of \equref{eq:ourS}, we place on the SLM a phase ramp whose frequency and orientation matches the translation $\pt$ we want to implement. We capture $K=3$ images of this phase ramp plus a global phasor $\phi_k \in \{0,\frac{\pi}{3},\frac{2\pi}{3}\}$.
Since the SLM modulates only part of the wave, we obtain the  measurement
\BE \hat{S}^{\inp^n,k}_t = \abs{u^{\inp^n}(\snsp) + e^{j\phi_k} u^{\inp^n}(\snsp+ \pt)}^2,\EE
where $k$ index the phase shift, $n$ index the fluorescent source, and $t$ index the translation  $\pt$ of the current measurement.
Using phase-shifting interferometry~\cite{hariharan1987digital}, we can extract the interference signal desired in \equref{eq:ourS} as:
\BE\label{eq:ourSint}
S^{\inp^n}_t= \sum_{k} e^{j\phi_k}\abs{u^{\inp^n}(\snsp) + e^{j\phi_k} u^{\inp^n}(\snsp+ \pt)}^2  = u^{\inp^n}(\snsp)u^{\inp^n}(\snsp+ \pt)^*,
\EE
In the presence of multiple incoherent sources, emission from different  sources do not interfere.
Thus, the measured intensity in each shot is equivalent to $ \hat{I}_t^k=\sum_n \hat{S}^{\inp^n,k}_t$.
With the phase shifting interferometry in \equref{eq:ourSint} we extract
$I_t=\sum_n S^{\inp^n}_t $.

\subsection{Setup for the spatial incoherent target}
We wanted to test our algorithm on incoherent illumination layout of arbitrary complexity, as in \figref{fig:showlaser}. Thus, in addition to fluorescent beads we created incoherent illumination patterns with a shifting laser setup using the setup illustrated in \figref{fig:setupfull}.
 For this, we imaged the diffused output of a fibered laser ($635nm$) to generate a point focused exactly at the back of the tissue (we use the validation camera for accurate focusing). The main camera at the other side of the tissue captured  the intensity scattered from this spot. We then translate the fiber output on a programmable xy stage to generate spots at other positions behind the tissue. We capture a sequence of images at each source position and sum their {\em intensities}, thus simulating incoherent summation from multiple sources. This setup allows us to control the layout and complexity of the sources, which is useful for analyzing our algorithm with patterns of controlled complexities as in \figref{fig:autocorrelation} and replicating the same layout behind multiple tissue slices as in \figref{fig:compScontrast} and \ref{fig:compthickrecon}.

\begin{figure*}[t!]
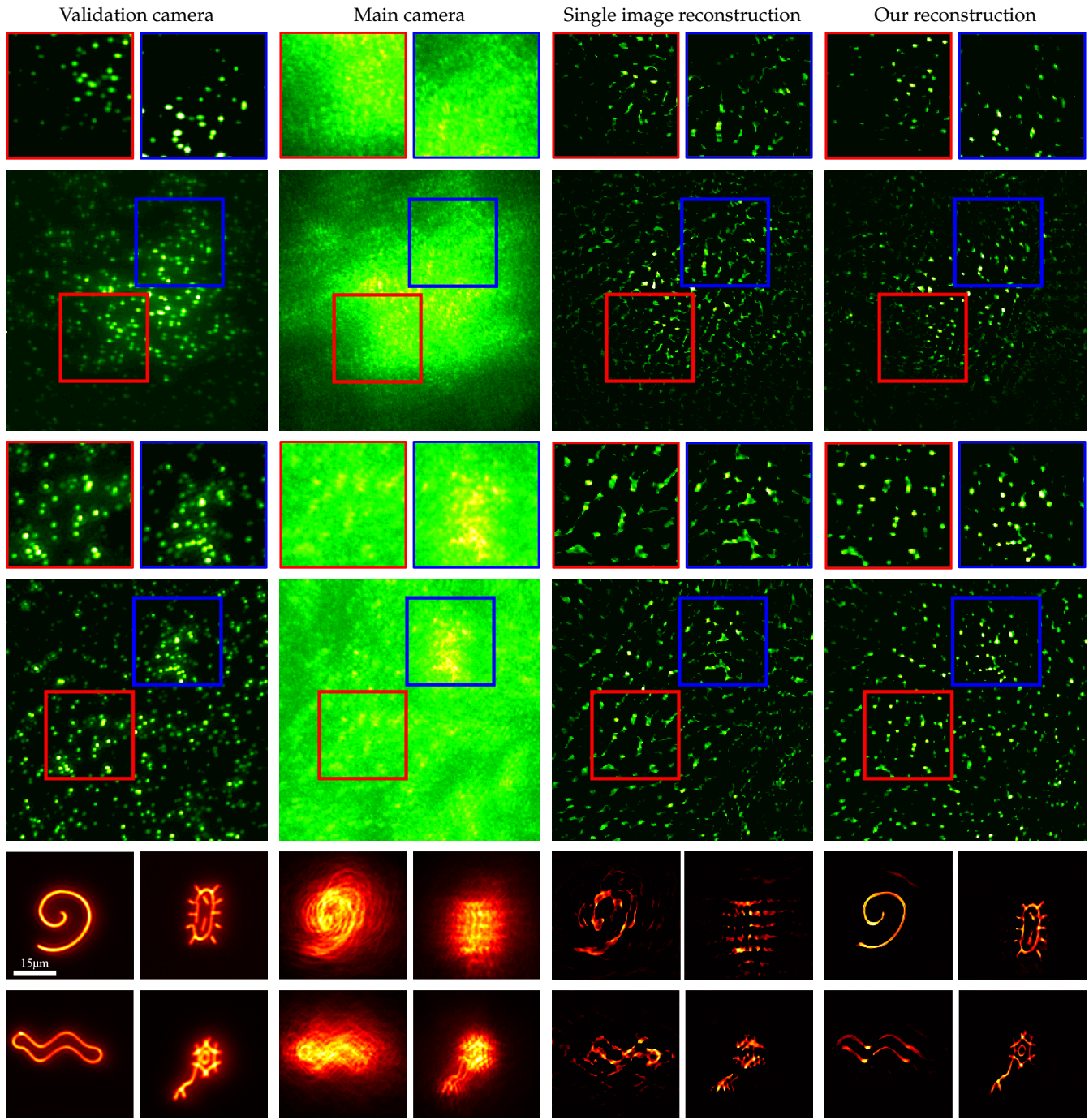

\centering
\begin{tabular}{*4{X}@{}}
Validation camera & Main camera & Single image reconstruction & Our reconstruction\\
\comparepics{./figure/fluorescent/150um_dense3_redo_part} \\
\comparepics{./figure/fluorescent/150um_dense3_redo_marked}\\
\comparepics{./figure/fluorescent/150um_dense4_redo_part} \\
\comparepics{./figure/fluorescent/150um_dense4_redo_marked}\\
\comparepics{./figure/nearfield/150um/mix_less}\\
\comparepics{./figure/nearfield/150um/mix_less2}\\
\end{tabular}
\caption{\textbf{Additional reconstruction results.} Top panel: results on a fluorescent bead image. Lower panel: results using  shifting laser inputs. Both inputs were images through  \textasciitilde$150\mu m$-thick tissue slices }
\label{fig:showadditional}
\end{figure*}

\subsection{Additional results}
In \figref{fig:showadditional} we show additional reconstruction results using fluorescent bead targets as well as shifting laser dots targets. Both used tissue slices of thickness around \textasciitilde$150{\mu}m$.
In \figref{fig:showparafilm} we use fluorescent beads behind a parafilm phantom, see characterization in \secref{sec:characterization}.

\subsection{Translating interferometry without phase correction}
As noted in \clmref{claim:trans-interf-auto-corr}, when using our \AlgName, we need a phase ramp correction to the auto-correlation of speckles. To support the claim, we compare results with and without the correction. In \figref{fig:autocorrelationfull}, we show the auto-correlation obtained by averaging  \AlgName measurements $I_t$ (\equref{eq:intef-It}) with and without the phase ramp correction of \equref{eq:avg-auto-corr-ramp}. The phase correction further improves the contrast, especially at larger displacements $\Dl$. In \figref{fig:compalgfull} we further show the reconstruction result of \AlgName with and without the phase ramp correction.
If the phase correction is not used not all the sources are recovered.
\begin{figure*}[t]
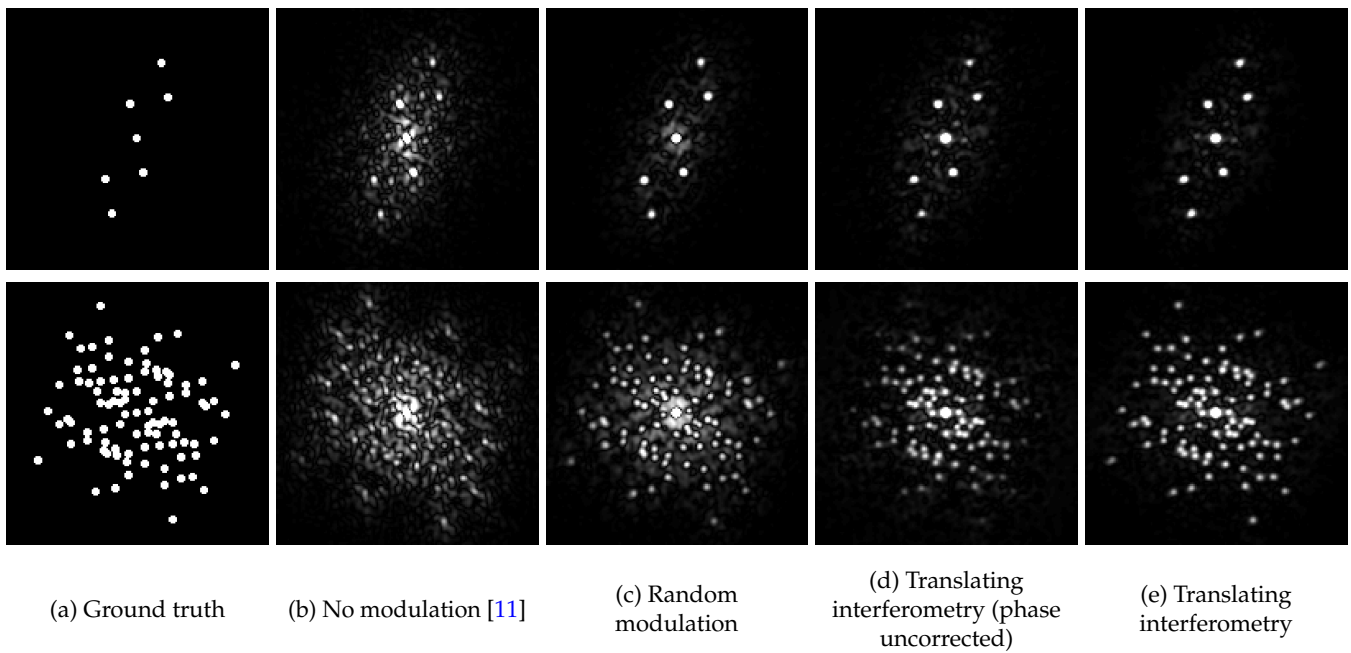

\centering
\begin{tabular}{*5{L}@{}}
\comparecorrfull{./figure/compare/collisionfree_fewer_A} \\
\comparecorrfull{./figure/compare/collisionfree_A} \\
(a) Ground truth & (b) No modulation \cite{SeeThroughSubmission}& (c) Random modulation
    & (d) Translating interferometry (phase uncorrected) & (e) Translating interferometry\\
\end{tabular}
\caption{\label{fig:autocorrelationfull}\textbf{Comparing speckle auto-correlation in \AlgName with and without correction.} (a) Input. (b) Auto-correlation from a single speckle image with no modulation. (c) Random phase mask in the Fourier plane. (d) \AlgNameC measurements without correction. The  larger displacements at the outer image regions are degraded due to phase variations. (e) After phase correction, {\AlgName} can recover further displacements.    }
\end{figure*} 
\begin{figure*}[!t]
\centering
\begin{tabular}{m{0.7em}*5{L}@{}}
&  \makecell{Validation\\ camera} & No modulation & Random modulation
 & \makecell{Translating\\ interferometry \\(phase uncorrected)} & \makecell{Translating\\ interferometry}\\
\comparealgfull{./figure/compare/collisionfree_ffrecon}{./figure/compare/collisionfree_recon} \\
\comparealgfull{./figure/compare/fluorescent_ffrecon_redo}{./figure/compare/fluorescent_recon_redo}
\end{tabular}
\caption{\textbf{Comparing \AlgName with and without phase correction.} In addition to \figref{fig:compalg}, we  compare \AlgName without phase correction.  Top part: sparse, simple shifting laser target. Lower part: challenging fluorescent beads target. For the simple target on the top, the full-frame algorithm using the auto-correlation of the {\AlgName} measurements without phase correction fails to recover the further beads. The local correlation approach which is more robust to noise can recover the target with and without correlation.  For the challenging target at the lower part, the full-frame algorithm fails completely even with phase correction. The local correlations algorithm can reconstruct the target using {\AlgName} modulations. However, without phase correction, it can only reconstruct a subset of the beads.}
\label{fig:compalgfull}
\end{figure*} 
\begin{figure}[!t]
\centering
\hspace{-2em}
\begin{tabular}{m{1.5em}SSS}
&\textasciitilde$100\mu m$ &\textasciitilde$150\mu m$ &\textasciitilde$200\mu m$\\
\textrot{\makecell{Input from \\ main camera}} &\comparethickness{speckle} \\
\textrot{\makecell{Our \\ reconstruction}} &\comparethickness{average_recon}
\end{tabular}
\includegraphics[width=18em]{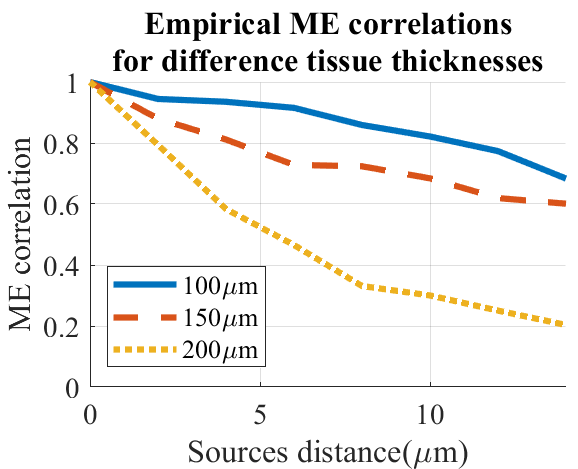}
\caption{\textbf{Compare reconstructions with different tissue thickness.}
Top panel: We use the shifting laser acquisition to capture the same illumination pattern through three different tissue slices of increasing thicknesses. As tissue thickness increase the spread of the speckles is wider and the statistical correlations are weaker. With  the \textasciitilde$200{\mu}m$-thick tissue the correlations are too wreak and reconstruction fails.
Lower panel: As the shifting laser setup allows capturing speckles from individual sources independently, we can empirically compute the decay of speckle correlation as a function of the distance between the sources. Indeed for thicker tissue slices the correlation decays faster as a function of source separation. This explains the reconstruction failure for the \textasciitilde$200{\mu}m$-thick tissue example.
}
\label{fig:compthickrecon}
\end{figure}

\subsection{Tissue thickness}
In \figref{fig:compthickrecon} we used the scanning laser setup to create the same illumination layout behind different tissue slices. This allows us to compare reconstructions through different tissue thicknesses.
For the thickest layer the reconstruction failed.

As discussed in~ \cite{SeeThroughSubmission}, thicker tissue leads to larger speckle spread and weaker memory effect correlation.
We used the fact that the scanning laser setup captures speckle patterns by different point sources separately, to evaluate statistical correlation through the different tissue layers.
In the top part of \figref{fig:compthickrecon} we plot the correlation we measured between speckle patterns generated by different sources, as a function of the displacement between the sources. As expected, as the tissue thickness  increases the correlation decays faster as a function of displacement, explaining the reconstruction failure in the lowest row.

These results demonstrate the limitation of our method: while we can improve correlation contrast, our approach is still based on the existence of some memory effect correlation and will fail when this correlation is too weak.

\subsection{Tissue preparation and characterization}\label{sec:characterization}
Most results in this paper used chicken breast tissue as a scattering material.
We cut thin slice from thawed chicken breast. To keep the tissue fresh, we did the fluorescent imaging experiment within 3 hours. The scanning laser targets require longer capture. To keep freshness  we squeeze the tissue between two cover glasses and seal them using nail polish.

As stated in \cite{schott2015characterization},  chicken breast tissue has an anisotropy parameter $g = 0.965$ and a mean free path (MFP) around $43.7\mu m$. However, these parameters may vary significantly between different tissue slices.

For a better characterization, we also use a parafilm phantom. This was calibrated in \cite{boniface2019noninvasive}, reporting an anisotropy parameter $g = 0.77$ and a MFP around $170\mu m$. We imaged through one parafilm layer whose  thickness is about $120\mu m$. Results for this phantom are demonstrated  in \figref{fig:showparafilm}. Note that while the parafilm has a longer MFP, it also has a smaller anisotropy factor and in practice the speckle spreads of both parafilm and chicken breast are comparable.

\begin{figure*}[t!]
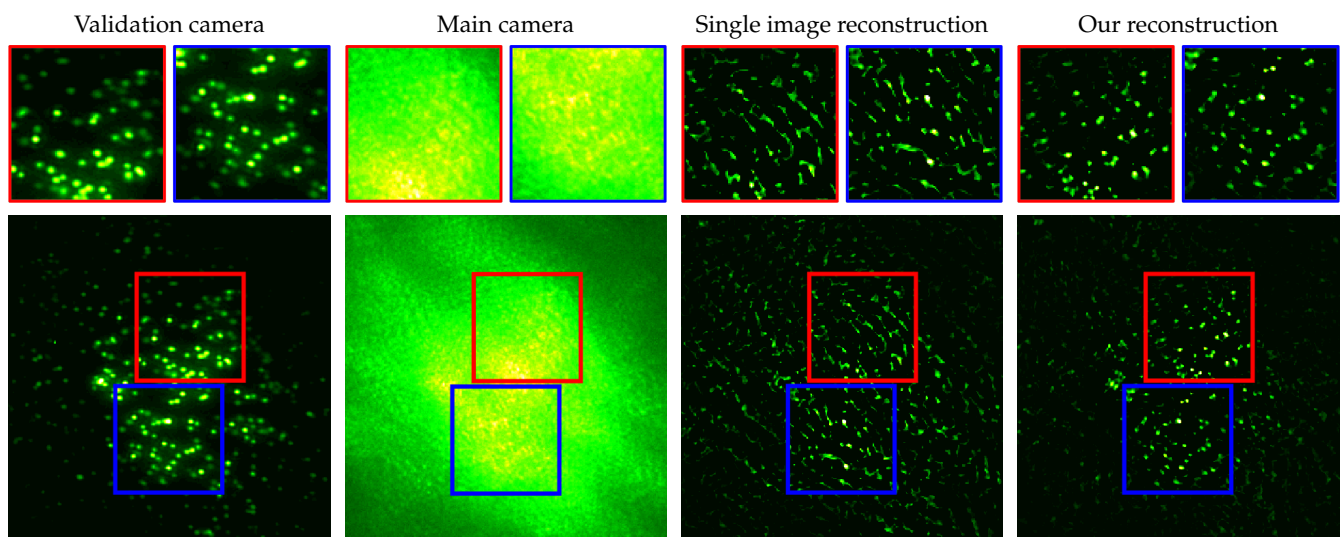

\centering
\begin{tabular}{*4{X}@{}}
Validation camera & Main camera & Single image reconstruction & Our reconstruction\\
\comparepics{./figure/fluorescent/parafilm_part} \\
\comparepics{./figure/fluorescent/parafilm_marked}\\
\end{tabular}
\caption{\textbf{Additional results.} reconstructing fluorescent beads behind a parafilm phantom. }
\label{fig:showparafilm}
\end{figure*}

\end{document}